\documentclass[11pt]{article}
\usepackage{amssymb,latexsym,amsmath,amsthm}
\usepackage[dvips]{graphicx}
\usepackage{cite}
\usepackage[only,llbracket,rrbracket]{stmaryrd} 
\usepackage{hyperref}
\newtheorem{theo}{Theorem}
\newtheorem{defi}[theo]{Definition}
\newtheorem{lemm}[theo]{Lemma}
\newtheorem{coro}[theo]{Corollary}
\newtheorem{rema}[theo]{Remark}
\newtheorem{prop}[theo]{Proposition}
\numberwithin{theo}{section} 
\newcommand{\ii}{\mathrm{i}}

\newcommand{\mR}{\mathbb{R}}
\newcommand{\mE}{\mathbb{E}}
\newcommand{\ux}{\underline{x}}
\newcommand{\uD}{\underline{D}}
\newcommand{\cD}{\mathcal{D}}
\newcommand{\cA}{\mathcal{A}}
\numberwithin{equation}{section}

\allowdisplaybreaks

\begin{document}
\begin{center}
\noindent {\Large \bf
On the algebra of symmetries of Laplace and Dirac operators
} \\[5mm]
{\bf Hendrik De Bie${}^{1}$, Roy Oste${}^{2}$, Joris Van der Jeugt${}^{3}$}\\[3mm]
${}^{1}$Department of Mathematical Analysis, Faculty of Engineering and Architecture, Ghent University, Krijgslaan 281-S8, 9000 Gent, Belgium\\[3mm]
${}^{2,3}$Department of Applied Mathematics, Computer Science and Statistics, Faculty of Sciences, Ghent University, Krijgslaan 281-S9, 9000 Gent, Belgium.\\
\end{center}
E-mail: {\tt Hendrik.DeBie@UGent.be}; {\tt Roy.Oste@UGent.be}; {\tt Joris.VanderJeugt@UGent.be}

\begin{abstract}
We consider a generalization of the classical Laplace operator, which includes the Laplace--Dunkl operator defined in terms of the differential-difference operators associated with finite reflection groups called Dunkl operators. For this Laplace-like operator, we determine a set of symmetries commuting with it, which are generalized angular momentum operators, and we present the algebraic relations for the symmetry algebra. In this context, the generalized Dirac operator is then defined as a square root of our Laplace-like operator. We explicitly determine a family of graded operators which commute or anticommute with our Dirac-like operator depending on their degree. The algebra generated by these symmetry operators is shown to be a generalization of the standard angular momentum algebra and the recently defined higher rank Bannai--Ito algebra.
\end{abstract}


\setcounter{equation}{0}
\section{Introduction} \label{sec:Introduction}%

The study of solutions of the Laplace equation or of the Dirac
equation, in any context or setting, is a major topic of investigation. For that purpose, a crucial role is played by the symmetries of the Laplace operator $\Delta$ or of the Dirac operator $\underline D$, i.e., operators commuting with $\Delta$ or (anti)commuting with $\underline D$. The symmetries involved and the algebras they generate lead to topics such as separation of variables and special functions. Without claiming completeness, we refer the reader to~\cite{KH,Eastwood,DeBie&Genest&Vinet-2016,DeBie&Genest&Vinet-2016-2,Boyer,Miller}.

For this paper, the context we have in mind is that of Dunkl derivatives~\cite{1989_Dunkl_TransAmerMathSoc_311_167,2003_Rosler}, i.e., where the ordinary derivative $\frac{\partial}{\partial x_i}$ is replaced by a Dunkl derivative ${\cal D}_i$ in the expression of the Laplace or Dirac operator.
One often refers to these operators as the Laplace--Dunkl and the Dirac--Dunkl operator. 
The chief purpose of this paper is to determine the symmetries of the Laplace--Dunkl operator and of the Dirac--Dunkl operator, and moreover study the algebra generated by these symmetries.

In the process of this investigation, it occurred to us that it is advantageous to treat this problem in a more general context, which we shall describe here in the introduction.
For this purpose, let us first turn to a standard topic in quantum mechanics: the description of the $N$-dimensional (isotropic) harmonic oscillator.
The Hamiltonian $\hat H$ of this oscillator (with the common convention $m=\omega=\hbar=1$ for mass, frequency and the reduced Planck constant) is given by
\begin{equation}
\hat H= \frac12 \sum_{j=1}^N \hat p_j^2 + \frac12 \sum_{j=1}^N \hat x_j^2.
\label{Ham}
\end{equation}
In canonical quantum mechanics, the coordinate operators $\hat x_j$ and momentum operators $\hat p_j$ are required to be (essentially self-adjoint) operators satisfying the canonical commutation relations
\begin{equation}
[\hat x_i, \hat x_j] =0, \qquad [\hat p_i, \hat p_j] =0, \qquad [\hat x_i, \hat p_j] = \ii\delta_{ij}.
\label{CCR}
\end{equation}
So in the ``coordinate representation,'' where $\hat x_j$ is represented by multiplication with the variable $x_j$, the operator $\hat p_j$ is represented by $\hat p_j= -\ii \frac{\partial}{\partial x_j}$. 

Because the canonical commutation relations are sometimes considered as ``unphysical'' or ``imposed without a physical motivation,'' more fundamental ways of quantization have been the topic of various research fields (such as geometrical quantization).
One of the pioneers of a more fundamental quantization procedure was Wigner, who introduced a method that later became known as ``Wigner quantization''~\cite{Palev1982b,PS1997,SV2005,VdJWigner,Wigner1950}.
Briefly said, in Wigner quantization one preserves all axioms of quantum mechanics, except that the canonical commutation relations are replaced by a more fundamental principle: the compatibility of the (classical) Hamilton equations with the Heisenberg equations of motion. 
Concretely, these compatibility conditions read
\begin{equation}
[\hat H,\hat x_j] = -\ii \hat p_j, \qquad [\hat H,\hat p_j] = \ii \hat x_j \qquad (j=1,\ldots,N).
\label{CC}
\end{equation}
Thus, for the quantum oscillator, one keeps the Hamiltonian~\eqref{Ham}, but replaces relations~\eqref{CCR} by~\eqref{CC}.
When the canonical commutation relations~\eqref{CCR} hold, the compatibility relations~\eqref{CC} are automatically valid (this is a version of the Ehrenfest theorem), but not vice versa.
Hence, Wigner quantization is a generalization of canonical quantization, and canonical quantization is just one possible solution of Wigner quantization.
Note that in Wigner quantization the coordinate operators $\hat x_j$ (and the momentum operators) in general do not commute, so this is of particular significance in the field of non-commutative quantum mechanics.

In a mathematical context, as in this paper, one usually replaces the physical momentum operator components $\hat p_j$ by operators $p_j=\ii \hat p_j$, and one also denotes the coordinate operators $\hat x_j$ by $x_j$. 
Then the operator $\hat H$ takes the form
\begin{equation}
H= -\frac12 \sum_{j=1}^N p_j^2 + \frac12 \sum_{j=1}^N x_j^2.
\label{Hpx}
\end{equation}
So in the canonical case, where $x_j$ stands for multiplication by the variable $x_j$, $p_j$ is just the derivative $\frac{\partial}{\partial x_j}$, and the first term of $H$ is (up to a factor $-1/2$) equal to the Laplace operator $\sum_{j=1}^N p_j^2= \sum_{j=1}^N \frac{\partial^2}{\partial x_j^2} = \Delta$.
In the more general case, the compatibility conditions~\eqref{CC} read
\begin{equation}
[H, x_j] = - p_j, \qquad [H, p_j] = - x_j \qquad (j=1,\ldots,N).
\label{CCm}
\end{equation}

We are now in a position to describe our problem in a general framework:
\begin{quote}
	{\em 
		Given $N$ commuting operators $x_1,\ldots,x_N$ and $N$ commuting operators $p_1,\ldots,p_N$, consider the operator 
		$H= \displaystyle -\frac12 \sum_{j=1}^N p_j^2 + \frac12 \sum_{j=1}^N x_j^2$, and suppose that the compatibility conditions~\eqref{CCm} hold. 
		Classify the symmetries of the generalized Laplace operator, i.e., classify the operators that commute with $\sum_{j=1}^N p_j^2$.
	}
\end{quote}
In other words, we are given $N$ operators $x_1,\ldots,x_N$ and $N$ operators $p_1,\ldots,p_N$ that satisfy
\begin{align}
& [x_i, x_j] =0, \qquad [p_i, p_j] =0, \label{cond1} \\*
& \bigg[\frac12 \sum_{i=1}^N p_i^2, x_j\bigg]=p_j, \qquad \bigg[\frac12 \sum_{i=1}^N x_i^2, p_j\bigg] = -x_j \label{cond2}.
\end{align}
Under these conditions, the {\em first problem} is: determine the operators that commute with the generalized Laplace operator
\begin{equation}
\Delta = \sum_{i=1}^N p_i^2.
\label{Laplace-op}
\end{equation}

Our two major examples of systems satisfying \eqref{cond1} and \eqref{cond2} are the ``canonical case'' and the ``Dunkl case.'' Another example, which we will not handle in detail, can be found in the context of discrete counterparts of the Laplacian~\cite{DRRS}. 

For the first example, $x_i$ is just the multiplication by the variable $x_i$, and $p_i$ is the derivative with respect to $x_i$: $p_i=\frac{\partial}{\partial x_i}$. Clearly, these operators satisfy \eqref{cond1} and \eqref{cond2}, and the operator $\Delta$ in~\eqref{Laplace-op} coincides with the 
classical Laplace operator.

For the second example, $x_i$ is again the multiplication by the variable $x_i$, but $p_i$ is the Dunkl derivative $p_i={\cal D}_i$, which is a certain differential-difference operator with an underlying reflection group determined by a root system (a precise definition follows later in this paper).
Conditions \eqref{cond1} still hold: the commutativity of the operators $x_i$ is trivial, but the commutativity of the operators $p_i$ is far from trivial~\cite{1989_Dunkl_TransAmerMathSoc_311_167,Etingof}.
Following~\cite{1989_Dunkl_TransAmerMathSoc_311_167}, also the conditions \eqref{cond2} are valid in the Dunkl case. 
The operator $\Delta$ in~\eqref{Laplace-op} now takes the form $\sum_{i=1}^N {\cal D}_i^2$ and is known as the Dunkl Laplacian or the Laplace--Dunkl operator.
By the way, it is no surprise that the operators $x_i$ and ${\cal D}_j$ do not satisfy the canonical commutation relations.
It is, however, very surprising that they satisfy the more general Wigner quantization relations (for a Hamiltonian of oscillator type).

So the solution of the first problem in the general context will in particular lead to the determination of symmetries of the Laplace--Dunkl operator.

Since we are dealing with these operators in an algebraic context, it is worthwhile to move to a closely related operator, the Dirac operator.
For this purpose, consider a set of $N$ generators of a Clifford algebra, i.e., $N$ elements $e_i$ satisfying 
\[
\{e_i,e_j\}=\epsilon\, 2\delta_{ij}
\]
where $\{a,b\} = ab + ba$ denotes the anticommutator and $\epsilon$ is $+1$ or $-1$. The generators $e_i$ are supposed to commute with the general operators $x_j$ and $p_j$.
Under the general conditions \eqref{cond1} and \eqref{cond2}, the {\em second problem} is now: determine the operators that commute (or anticommute) with the generalized Dirac operator
\begin{equation}
\underline D = \sum_{i=1}^N e_i p_i.
\label{Dirac-op}
\end{equation}
Obviously, this is a refinement of the first problem, since $\underline D^2= \epsilon\, \Delta$. 

For our two major examples, in the canonical case the operator~\eqref{Dirac-op} is just the classical Dirac operator; in the Dunkl case, the operator~\eqref{Dirac-op} is known as the Dirac--Dunkl operator. 

In the present paper, we solve both problems in the general framework \eqref{cond1}--\eqref{cond2}, and even go beyond it by determining the algebraic relations satisfied by the symmetries.
In section~\ref{sec:2} we consider the generalized Laplace operator $\Delta$ and determine all symmetries, i.e., all operators commuting with $\Delta$ (Theorem~\ref{theoSyms}). Next, in Theorem~\ref{theoLAlg} the algebraic relations satisfied by these symmetries are established.
For the generalized Dirac operator $\underline D$, the determination of the symmetries is computationally far more involved. In section~\ref{sec:3}, Theorem~\ref{TOA} classifies essentially all operators that commute or anticommute with $\underline D$. 
In the following subsections, we derive the quadratic relations satisfied by the symmetries of the Dirac operator. The computations of these relations are very intricate, and involve subtle techniques. 
Fortunately, there is a case to compare with. For the Dunkl case, in which the underlying reflection group is the simplest possible example (namely ${\mathbb Z}_2^N$), the symmetries and their algebraic relations have been determined in~\cite{DeBie&Genest&Vinet-2016,DeBie&Genest&Vinet-2016-2} and give rise to the so-called (higher rank) Bannai--Ito algebra. Our results can be considered as an extension of these relations to an arbitrary underlying reflection group, in fact in an even more general context.

\section{Symmetries of Laplace operators} \label{sec:2}%

We start by formally describing the operator algebra that will contain the desired symmetries of a generalized Laplace operator \eqref{Laplace-op}, as brought up in the introduction. 
\begin{defi}\label{defA}
We define the algebra $\mathcal{A}$ to be the unital (over the field $\mR$ or $\mathbb{C}$) associative algebra generated by the $2N$ elements $x_1,\dotsc,x_N$ and $p_1,\dotsc,p_N$ subject to the following relations:
\begin{align*}
& [x_i, x_j] =0, \qquad\qquad\quad [p_i, p_j] =0, \\
& \bigg[\frac12 \sum_{i=1}^N p_i^2, x_j\bigg]=p_j, \qquad \bigg[\frac12 \sum_{i=1}^N x_i^2, p_j\bigg] = -x_j .
\end{align*}
\end{defi}
Note that an immediate consequence of the relations of $\mathcal{A}$ is
\[
[x_i,p_j]=[x_i,-[H,x_j]]=-[[x_i,H],x_j]=-[p_i,x_j]=[x_j,p_i],
\]
where $H$ is given by~\eqref{Hpx}. This reciprocity 
\begin{equation}
[x_i,p_j]=[x_j,p_i]
\label{cond2b}
\end{equation} 
will be useful for many ensuing calculations, starting with the following theorem.

\begin{theo}\label{theosl2}
The algebra $\cA$ contains a copy of the Lie algebra $\mathfrak{sl}(2)$ generated by the elements
 \begin{equation}\label{Euler}
\frac{|x|^2}{2}  = \frac12\sum_{i=1}^N x_i^2 , \qquad -\frac{\Delta}{2} =  -\frac12\sum_{i=1}^N p_i^2,\qquad \mathbb{E} =  \frac12 \sum_{i=1}^N \{p_i,x_i\},
 \end{equation}
satisfying the relations
\[
\Big[ \mE , \frac{|x|^2}{2} \Big] = |x|^2, \qquad  \Big[ \mE , -\frac{\Delta}{2} \Big] = \Delta, \qquad \Big[  \frac{|x|^2}{2} , -\frac{\Delta}{2} \Big] = \mE.
\]
\end{theo}
\begin{proof}
By direct computation we have	
 \begin{equation*}
\frac14[\Delta,|x|^2] =  \frac14\sum_{i=1}^N [\Delta,x_i^2] =  \frac14\sum_{i=1}^N \{[\Delta,x_i],x_i\} = \frac12 \sum_{i=1}^N \{p_i,x_i\}.
 \end{equation*}
Using the commutativity of $p_1,\dotsc,p_N$ and relation~\eqref{cond2b}, we have
 \begin{align*}
 [\mE ,\Delta ] =\ & \frac12\sum_{i=1}^N\sum_{j=1}^N [ \{p_i,x_i\},p_j^2]
 =\  \frac12\sum_{i=1}^N\sum_{j=1}^N \big\{[ \{p_i,x_i\},p_j] ,p_j\big\}\\
 =\ &  \frac12\sum_{i=1}^N\sum_{j=1}^N \big\{\big(p_i(x_ip_j-p_jx_i)+(x_ip_j-p_jx_i)p_i\big) ,p_j\big\} \\
 =\ & \frac12\sum_{i=1}^N\sum_{j=1}^N \big\{\big(p_i(x_jp_i-p_ix_j)+(x_jp_i-p_ix_j)p_i\big),p_j\big\}
 \\
 =\ & -\frac12 \sum_{j=1}^N \bigg\{\bigg[ \sum_{i=1}^N p_i^2,x_j\bigg],p_j\bigg\}
 = - \sum_{j=1}^N \{p_j,p_j\} = -2\Delta.
 \end{align*}
 In the same manner, using now the commutativity of $x_1,\dotsc,x_N$, we find $[\mE,|x|^2] =2|x|^2$.  		
\end{proof}
In the spirit of Howe duality~\cite{Howe,Howe2}, our objective is to determine the subalgebra of $\cA$ which commutes with the Lie algebra $\mathfrak{sl}(2)$ realized by
$\Delta$ and $|x|^2$ as appearing in Theorem~\ref{theosl2}. As mentioned in the introduction, the element $\Delta$ corresponds to a generalized version of the Laplace operator, which reduces to the classical Laplace operator for a specific choice of the elements $p_1,\dotsc,p_N$. In the (Euclidean) coordinate representation, $|x|^2$ of course represents the norm squared.

\subsection{Symmetries} \label{sec:2.2}%

As $p_1,\dotsc,p_N$ are commuting operators, by definition they also commute with $\Delta$. However, in general they are not symmetries of $|x|^2$. An immediate first example of an operator which does commute with both $\Delta$ and $|x|^2$ is given by the Casimir operator (in the universal enveloping algebra) of their $\mathfrak{sl}(2)$ realization
\begin{equation}\label{Omega}
\Omega =\mE^2 -2\mE-|x|^2\Delta  \in \mathcal{U}( \mathfrak{sl}(2)) \subset \cA.
\end{equation}
Note that this operator is of the same order in both $x_1,\dotsc,x_N$ and $p_1,\dotsc,p_N$ as it has to commute with both $\Delta$ and $|x|^2$. More precisely it is of fourth order in the generators of $\cA$, being quadratic in $x_1,\dotsc,x_N$ and quadratic in $p_1,\dotsc,p_N$.
We now set out to consider the most elementary symmetries, those which are of second order in the generators of $\cA$. 

\begin{theo}\label{theoSyms}
	In the algebra $\cA$, the elements which are quadratic in the generators $x_1,\dotsc,x_N$ and $p_1,\dotsc,p_N$, and which commute with $\Delta$ and $|x|^2$ are spanned by
	\begin{equation}\label{syms}
	L_{ij} = x_i p_j -x_j p_i, \qquad C_{ij} = [p_i,x_j] = p_i x_j -x_j p_i\qquad (i,j \in \{1,\dotsc,N\}).
	\end{equation}
\end{theo}
Note that when $i=j$ we have $L_{ii} = 0$, while $C_{ii} = [p_i,x_i]$ does not necessarily vanish. Moreover, as $L_{ij} = -L_{ji}$, every symmetry $L_{ij}$ is up to a sign equal to one of the $N(N-1)/2$ symmetries $\{L_{ij} \mid 1\leq i<j\leq N\}$. By relation~\eqref{cond2b}, $C_{ij} = C_{ji}$, and thus, we have $N(N+1)/2$ symmetries $\{C_{ij} \mid 1\leq i\leq j\leq N\}$. 
In total, this gives $N^2$ generically distinct symmetries.
\begin{proof}
	We first show that $L_{ij}$ and $C_{ij}$ as defined by~\eqref{syms} are indeed symmetries of $\Delta$ and $|x|^2$. 
As $\Delta$ commutes with $p_1,\dotsc,p_N$ and using condition \eqref{cond2}, we have for $i,j \in \{1,\dotsc,N\}$
\[
[\Delta,x_ip_j - x_j p_i]   
= x_i[\Delta,p_j ] + [\Delta,x_i ]p_j - x_j[\Delta,  p_i] - [\Delta, x_j ]p_i 
=  2p_i p_j  - 2p_j p_i
= 0.
\]
In the same manner, we have $[\Delta,p_ix_j - x_j p_i]=0$. The relations for $|x|^2$ follow similarly.

Now, a general element $S\in\cA$ quadratic in the generators $x_1,\dotsc,x_N$ and $p_1,\dotsc,p_N$ is of the form
\[
S = \sum_{i,j} ( a_{ij} x_i p_j + b_{ij} p_i x_j + c_{ij} x_ix_j + d_{ij}  p_ip_j)\rlap{\,,}
\]
where $i$ and $j$ are summed over $\{1,\dotsc,N\}$ and $a_{ij}, b_{ij} , c_{ij}, d_{ij}$ are scalars. 
Using relations~\eqref{cond1} and \eqref{cond2}, we have 
\begin{align*}
	&	\frac12[\Delta,S] = \sum_{i,j} ( a_{ij} p_i p_j + b_{ij} p_i p_j + c_{ij}(x_ip_j+p_ix_j ) ) \\ 
	=\	&  \sum_{i,j}  c_{ij}(x_ip_j+p_ix_j )  + \sum_{i} (a_{ii}+ b_{ii})p_i^2   + \sum_{i<j}  (a_{ij} + a_{ji}  + b_{ij}  + b_{ji} )p_i p_j   \rlap{\,.}
\end{align*}	
In order for this to vanish, the coefficients must satisfy $a_{ii}+ b_{ii}=0$, $c_{ij}=0$ for all $i,j$, and 
$a_{ij} + a_{ji}  + b_{ij}  + b_{ji}=0$ for $i<j$ (though by symmetry also for $i>j$, and by the previous case also for $i=j$, thus finally for all $i,j$). The condition $[|x|^2,S]=0$ yields the additional requirements $d_{ij}=0$ for all $i,j$. Hence, say $c_{ij}=0$, $d_{ij}=0$, $a_{ii}=- b_{ii}$, and $a_{ji} =-a_{ij} - b_{ij}  - b_{ji}$, then the symmetry $S$ is of the form
\[
S = \sum_{i} b_{ii}(p_ix_i-x_ip_i)  + \sum_{i<j} ( a_{ij}( x_i p_j -x_j p_i ) + b_{ij} (p_i x_j-x_j p_i)  + b_{ji} (p_j x_i-x_j p_i) )
\]
where in the right-hand side we recognize $C_{ii}$, $L_{ij}$, $C_{ij}$ and $C_{ji}+L_{ij}$.
\end{proof}
	
\subsection{Symmetry algebra} \label{ssec:Algebra}%

For the following results, we make explicit use of the symmetry $C_{ij}=[p_i,x_j]$ being symmetric in its two indices, by relation~\eqref{cond2b}. 
This is the case for $p_i$ corresponding to classical partial derivatives, but also for their generalization in the form of Dunkl operators. We will return in detail to these examples in section~\ref{ssec:Examples}. Another consequence of relation~\eqref{cond2b} pertains to the form of the other symmetries of Theorem~\ref{theoSyms}. By means of $x_ip_j-p_jx_i=x_jp_i-p_ix_j$ we readily observe that
	\begin{equation}\label{Lij}
L_{ij}=x_ip_j-x_jp_i = p_jx_i-p_ix_j.
	\end{equation}
Given these symmetry properties, the symmetries of Theorem~\ref{theoSyms} generate an algebraic structure within $\cA$ whose relations we present after the following lemma. 
\begin{lemm}\label{lemma2}	
In the algebra $\cA$, the symmetries~\eqref{syms} satisfy the following relations for all $i,j,k\in\{1,\dotsc,N\}$ 
\[
	[C_{ij},p_k] =[C_{kj},p_i],
\qquad\mbox{
and }
\qquad
[C_{ij},x_k] =[C_{kj},x_i].
\]
Moreover, we also have 
\[
	L_{ij}p_k + L_{ki}p_j + L_{jk}p_i = 0 = p_kL_{ij} + p_jL_{ki} + p_iL_{jk},
\]
and
\[
 x_kL_{ij} + x_jL_{ki} + x_iL_{jk}= 0 =	L_{ij}x_k + L_{ki}x_j + L_{jk}x_i .
\]
\end{lemm}
\begin{proof}
For the first relation, writing out the commutators in $\big[[p_k ,x_j ],p_i\big] -\big[[p_i , x_j ],p_k\big]$ we find
\[
	p_k x_j p_i -x_j p_k p_i - p_ip_k x_j + p_ix_jp_k  - p_i  x_j p_k+ x_jp_i   p_k+p_kp_i  x_j -p_k  x_jp_i.
\]
We see that all terms cancel due to the mutual commutativity of the operators $p_1,\dotsc,p_N$. The other relation of the first line follows in the same way.

For the other two relations, the identities follow immediately by choosing the appropriate expression for $L_{ij}$ of \eqref{Lij} and making use of the commutativity of either $x_1,\dotsc,x_N$ or $p_1,\dotsc,p_N$.
\end{proof}

\begin{theo}\label{theoLAlg}
		In the algebra $\cA$, the symmetries~\eqref{syms} satisfy the following relations for all $i,j,k,l\in\{1,\dotsc,N\}$, 
	\begin{align}\label{rel1}
	[ L_{ij}, L_{kl}] & =L_{il}C_{jk} +L_{jk}C_{il}+L_{ki}C_{lj}+L_{lj}C_{ki} \\
	& =C_{jk}L_{il} +C_{il}L_{jk}+C_{lj}L_{ki}+C_{ki}L_{lj}, \notag
	\end{align}
	\begin{equation}\label{rel2}
	\{L_{ij},L_{kl}\}+\{L_{ki},L_{jl}\} +\{L_{jk},L_{il}\}  =0\rlap{\,,}
\end{equation}
\begin{equation}\label{rel3}
	[L_{ij} ,C_{kl}] +[L_{ki},C_{jl}]+[L_{jk},C_{il}] =0\rlap{\,,}
\end{equation}
and
\begin{equation}\label{rel4}
	L_{ij}L_{kl}+L_{ki}L_{jl} +L_{jk}L_{il}  =	L_{ij}C_{kl} +L_{ki}C_{jl}+L_{jk}C_{il} \rlap{\,.}
\end{equation}
\end{theo}
\begin{proof}
	We will prove the first line of the first relation, i.e.~\eqref{rel1}, the second line follows in a similar manner. We have 
	\begin{align*}
	[x_ip_j - x_j p_i,x_kp_l - x_l p_k]   = \ &    [x_ip_j ,x_kp_l ] - [x_ip_j , x_l p_k]  -[ x_j p_i,x_kp_l ]  + [ x_j p_i, x_l p_k]  \\
	 =\ &  x_i[p_j ,x_k ]p_l + x_k[x_i ,p_l ]p_j - x_i[p_j , x_l ]p_k- x_l[x_i ,  p_k]p_j  \\
	 & -x_j[  p_i,x_k ]p_l -x_k[ x_j,p_l ] p_i  + x_j[  p_i, x_l ]p_k + x_l[ x_j ,  p_k]p_i\\
	 =\ &  x_iC_{jk}p_l - x_kC_{li}p_j - x_iC_{jl}p_k+ x_lC_{ki}p_j  \\
	 & -x_jC_{ik}p_l +x_kC_{lj} p_i  + x_jC_{il}p_k - x_lC_{kj}p_i.
	 \end{align*}
	 Swapping all operators $p_l$ with $C_{jk}$, we find
	 	\begin{align*}
	  [ L_{ij}, L_{kl}] =\ &  
	  x_ip_lC_{jk} + x_i[C_{jk},p_l] 
	  - x_kp_jC_{li}  - x_k[C_{li},p_j ]
	   - x_ip_kC_{jl} - x_i[C_{jl},p_k] \\
	 & + x_lp_jC_{ki}   + x_l[C_{ki},p_j ] 
	 -x_jp_lC_{ik}  -x_j[C_{ik},p_l] 
	 +x_kp_iC_{lj}   \\
	 & +x_k[C_{lj} ,p_i ]
	 + x_jp_kC_{il}   + x_j[C_{il},p_k ] 
	 - x_lp_iC_{kj} - x_l[C_{kj},p_i ] \\
	 =\ &  x_ip_l C_{jk}   - x_kp_jC_{li}  - x_ip_kC_{jl}   + x_lp_jC_{ki}  -x_jp_lC_{ik}	 +x_kp_iC_{lj} + x_jp_kC_{il} - x_lp_iC_{kj} \\
	 & + x_i\big([C_{jk},p_l] -[C_{jl},p_k]\big)
	 + x_k\big(-[C_{il},p_j ]+[C_{lj} ,p_i ]\big)\\
	& - x_l\big(-[C_{ki},p_j ] + [C_{kj},p_i ]\big) 
	  -x_j\big([C_{ik},p_l]  -[C_{il},p_k ] \big)
	    \\
	  =\ &  L_{il}C_{jk} +L_{jk}C_{il}+L_{ki}C_{lj}+L_{lj}C_{ki}, 
	\end{align*}
	where we used $C_{jk}=C_{kj}$ and Lemma~\ref{lemma2}. 
	
The identities \eqref{rel2} and \eqref{rel3} follow by making explicit use of both expressions of \eqref{Lij} for $L_{ij}$. For the left-hand side of \eqref{rel2} we have
\begin{align*}
	&	L_{ij}L_{kl}+L_{kl}L_{ij}+L_{ki}L_{jl}+L_{jl}L_{ki} +L_{jk}L_{il} +L_{il}L_{jk}  \\
	= \ & (x_ip_j-x_jp_i) (p_lx_k-p_kx_l)+ ( x_kp_l-x_lp_k)(p_jx_i-p_ix_j)
		+ (x_kp_i-x_ip_k)(p_lx_j-p_jx_l)\\
	&+(x_jp_l-x_lp_j)(p_ix_k-p_kx_i)
		+(x_jp_k-x_kp_j)(p_lx_i-p_ix_l)+(x_ip_l-x_lp_i)(p_kx_j-p_jx_k),
\end{align*}
where one observes that all terms vanish due to the commutativity of  $p_1,\dotsc,p_N$.
	
Working out the commutators, the left-hand side of \eqref{rel3} becomes
\begin{align*}
	 & 	 L_{ij} [p_l,x_k]-  [p_l,x_k]L_{ij} +
	L_{ki}[p_l,x_j]-[p_l,x_j]L_{ki}+L_{jk}[p_l,x_i]-[p_l,x_i]L_{jk} \\
	= \ &  L_{ij} p_lx_k-  p_lx_kL_{ij} +
	L_{ki}p_lx_j-p_lx_jL_{ki}+L_{jk}p_lx_i-p_lx_iL_{jk} \\
	&  -L_{ij} x_kp_l+  x_kp_lL_{ij} -
	L_{ki}x_jp_l+x_jp_lL_{ki}-L_{jk}x_ip_l+x_ip_lL_{jk} .
\end{align*}
 Hence, plugging in suitable choices for the symmetries $L_{ij}$, this becomes
\begin{align*}
	&(x_ip_j-x_jp_i) p_lx_k-  p_lx_k(x_ip_j-x_jp_i)+
	(x_kp_i-x_ip_k)p_lx_j -p_lx_j(x_kp_i-x_ip_k)\\ & +(x_jp_k-x_kp_j)p_lx_i-p_lx_i(x_jp_k-x_kp_j) -(p_jx_i-p_ix_j) x_kp_l+  x_kp_l(p_jx_i-p_ix_j)\\ & -
	(p_ix_k-p_kx_i)x_jp_l +x_jp_l(p_ix_k-p_kx_i)-(p_kx_j-p_jx_k)x_ip_l+x_ip_l(p_kx_j-p_jx_k) .
\end{align*}
One observes that all terms vanish due to the commutativity of $x_1,\dotsc,x_N$ and $p_1,\dotsc,p_N$ respectively.

For the final relation, using the definitions~\eqref{syms} we have 
\begin{align*}\label{rel4}
	&	L_{ij}(L_{kl}-C_{kl})+L_{ki}(L_{jl}-C_{jl}) +L_{jk}(L_{il}-C_{il})  \\
	= \ & L_{ij}(x_kp_l-p_kx_l)+L_{ki}(x_jp_l-p_jx_l) +L_{jk}(x_ip_l-p_ix_l)\\
	= \ & 	(	L_{ij}x_k + L_{ki}x_j + L_{jk}x_i)p_l - (L_{ij}p_k + L_{ki}p_j + L_{jk}p_i )x_l
	\rlap{\,,}
\end{align*}
which vanishes by Lemma~\ref{lemma2}. 
\end{proof}

\subsection{Examples}\label{ssec:Examples}

\noindent\textbf{Example 2.1.} As a first example, we consider $N$ mutually commuting variables $x_1,\dotsc,x_N$, doubling as operators acting on functions by left multiplication with the respective variable and $p_j$ being just the derivative $\partial/\partial x_j$ for $j\in\{1,\dotsc,N\}$. In this case, obviously $p_1,\dotsc,p_N$ mutually commute and the operators of interest are
\[
\Delta = \sum_{i=1}^N \frac{\partial^2}{\partial x_i^2} ,\qquad  |x|^2 =   \sum_{i=1}^N x_i^2,\qquad H = -\frac12 \Delta + \frac12 |x|^2,
\] 
which satisfy
\[
\frac12[ \Delta,x_i] = \frac{\partial}{\partial x_i} = p_i,\qquad \frac12\big[ |x|^2,p_i\big] = -x_i.
\]

By Theorem~\ref{theoSyms}, we have the following symmetries: 
\[  
L_{ij} = x_i  \frac{\partial}{\partial x_j} -x_j  \frac{\partial}{\partial x_i}, \qquad C_{ij} = \delta_{ij} = \begin{cases}
1 & \text{if }i=j \\
0 & \text{if }i\neq j.
\end{cases}
\]
While $C_{ij}$ is a scalar for every $i,j$, the $L_{ij}$ symmetries are the standard angular momentum operators whose symmetry algebra is the Lie algebra $\mathfrak{so}(N)$:
	\[
	[ L_{ij}, L_{kl}]  =L_{il}\delta_{jk} +L_{jk}\delta_{il}+L_{ki}\delta_{lj}+L_{lj}\delta_{ik} .
	\]
This is in accordance with Theorem~\ref{theoLAlg} as in this case evidently $C_{ij} = C_{ji}$.

Note that $\Delta$ and $|x|^2$ are also invariant under $\mathrm{O}(N)$, the group of orthogonal transformations on $\mR^{N}$, but these transformations are not contained in the algebra $\cA$.

\noindent\textbf{Example 2.2.} A more intriguing example is given by a generalization of partial derivatives to dif\-fer\-en\-tial-dif\-fer\-ence operators associated with a Coxeter or Weyl group $W$. Let $R$ be a (reduced) root system and $k$ a multiplicity function which is invariant under the natural action of the Weyl group $W$ consisting of all reflections associated with $R$, 
\[
\sigma_{\alpha}(x)={x} -2\langle {x},{\alpha}\rangle{\alpha}/\|{\alpha}\|^2 ,\qquad \alpha \in R, x\in\mathbb{R}^N.
\]
For $\xi\in\mathbb{R}^N$, the Dunkl operator~\cite{1989_Dunkl_TransAmerMathSoc_311_167,2003_Rosler} is defined as
\[
\mathcal{D}_{ \xi}f(x) := \frac{\partial}{\partial\xi} f(x) + \sum_{\alpha\in R_+} k(\alpha)  \frac{f(x)-f(\sigma_{\alpha}(x))}{\langle\alpha,x\rangle}\langle\alpha,\xi\rangle ,
\]	
where the summation is taken over all roots in $R_+$, a fixed positive subsystem of $R$. 
For a fixed root system and function $k$, the Dunkl operators associated with any two vectors commute, see~\cite{1989_Dunkl_TransAmerMathSoc_311_167}. Hence, they form potential candidates for the operators $p_1,\dotsc,p_N$ satisfying condition~\eqref{cond1}. 
The operator of interest is the Laplace--Dunkl operator $\Delta_k$, which can be written as 
\[
\Delta_k = \sum_{i=1}^N (\cD_{\xi_i})^2
\]
for any orthonormal basis $\{\xi_1,\dotsc,\xi_N\}$ of $\mR^N$. For the orthonormal basis associated with the coordinates $x_1,\dotsc,x_N$, we use the notation
\begin{equation}\label{Dunkl}
\mathcal{D}_{ i}f(x) := \frac{\partial}{\partial x_i} f(x) + \sum_{\alpha\in R_+} k(\alpha)  \frac{f(x)-f(\sigma_{\alpha}(x))}{\langle\alpha,x\rangle} \alpha_i \qquad i\in\{1,\dotsc,N\} 
\end{equation}
where $\alpha_i = \langle \alpha , \xi_i \rangle$.

For our purpose, let $x_j$ again stand for multiplication by the variable $x_j$ and take now $p_j=\cD_j$ for $j\in\{1,\dotsc,N\}$. Besides condition~\eqref{cond1}, condition~\eqref{cond2} is also satisfied 
(see, for instance, \cite{1989_Dunkl_TransAmerMathSoc_311_167,2003_Rosler}). We note that the $\mathfrak{sl}(2)$ relations in this context were already obtained by \cite{Heckman-1991}.

By Theorem~\ref{theoSyms}, we have as symmetries, on the one hand, the Dunkl version of the angular momentum operators
\[
L_{ij} = x_i \cD_j - x_j \cD_i.
\]
On the other hand, the symmetries 
\[
 C_{ij} =
[\cD_i,x_j] = \delta_{ij} +  \sum_{\alpha\in R_+} 2k(\alpha) \alpha_i \alpha_j \sigma_{\alpha}
\] 
consist of linear combinations of the reflections in the Weyl group, with coefficients determined by the multiplicity function $k$ and the roots of the root system. This is of course in agreement with $\Delta_k$ being $W$-invariant~\cite{2003_Rosler}. The Weyl group is a subgroup of $\mathrm{O}(N)$, and in this case the algebra $\cA$ does contain these reflections in $W$.

Note that indeed $C_{ij} = C_{ji}$, in accordance with relation~\eqref{cond2b}. Theorem~\ref{theoLAlg} now yields the Dunkl version of the angular momentum algebra: 
\begin{align*}
&	[ L_{ij}, L_{kl}] =  L_{il}C_{jk} +L_{jk}C_{il}+L_{ki}C_{lj}+L_{lj}C_{ki} \\
&	= L_{il}\delta_{jk} +L_{jk}\delta_{il}+L_{ki}\delta_{lj}+L_{lj}\delta_{ki}  +  \sum_{\alpha\in R_+} 2k(\alpha) \big( L_{il} \alpha_j \alpha_k +  L_{jk} \alpha_i \alpha_l +  L_{ki} \alpha_l \alpha_j +  L_{lj} \alpha_k \alpha_i \big)  \sigma_{\alpha}.
\end{align*}
This relation states the interaction of the $L_{ij}$ symmetries among one another. 
The interaction between the symmetries $C_{kl}$ is governed by the group multiplication of the Weyl group, while the relations for the symmetries $L_{ij}$ and $C_{kl}$ follow immediately from the action of a reflection $\sigma_{\alpha}\in W$ on the coordinate variables and the Dunkl operators:
\[
\sigma_{\alpha}\, \xi = \sigma_{\alpha} (\xi) \,\sigma_{\alpha}
,\qquad  \sigma_{\alpha} \,\cD_{\xi} = \cD_{\sigma_{\alpha} (\xi)}\, \sigma_{\alpha}
\qquad (\xi\in\mathbb{R}^N)\rlap{\,.}
\]
Let $\{\xi_1,\dotsc,\xi_N\}$ denote the orthonormal basis associated with the coordinates $x_1,\dotsc,x_N$, then 
\[
\sigma_{\alpha} L_{ij} =  L_{\sigma_{\alpha} (\xi_i)\sigma_{\alpha} (\xi_j)} \sigma_{\alpha} ,
\]
where for $\xi,\eta\in\mathbb{R}^N$ we define
\[ 
L_{\xi\eta} = \langle x ,\xi\rangle \cD_{\eta} - \langle x ,\eta\rangle \cD_{\xi}  = \sum_{k,l} \langle\xi,\xi_k\rangle \langle\eta,\xi_l\rangle  L_{kl}  \rlap{\,.}
\]
This allows us to interchange any two symmetries of the form $L_{ij}$ and $C_{kl}$.

These results have been considered before, namely for the specific case $W=(\mathbb{Z}_2)^3$ in~\cite{2015_Genest&Vinet&Zhedanov_CommMathPhys_336_243}, and for $W=S_N$, and also for arbitrary Coxeter group, in~\cite{Feigin}. Furthermore, relation~\eqref{rel4} has been obtained already in the Dunkl case, and dubbed ``the crossing relation''~\cite{Feigin}.

\section{Symmetries of Dirac operators} \label{sec:3}

We now turn to a closely related operator of
the generalized Laplace operator considered in the preceding section, namely the Dirac operator. 
For an operator of the form~\eqref{Laplace-op}, one can construct a ``square root'' by introducing a set of elements $e_1,\dotsc,e_N$ which commute with $x_i$ and $p_i$ for all $i\in\{1,\dotsc,N\}$ and which satisfy the following relations
\begin{equation}
\label{Clifford}
\{e_i,e_j\} = e_ie_j +e_je_i = \epsilon\,2\delta_{ij},
\end{equation}
where $\epsilon = \pm 1$, or thus for $i\neq j$
\[
(e_i)^2 = \epsilon = \pm 1, \qquad   e_ie_j +e_je_i = 0.
\]
We use these elements to define the following two operators
\[
\uD = \sum_{i=1}^N e_i p_i , \qquad \ux = \sum_{i=1}^N e_i x_i,
\]
whose squares equal
\[
\uD^2 = \epsilon\sum_{i=1}^N (p_i)^2 = \epsilon \Delta , \qquad \ux^2 = \epsilon\sum_{i=1}^N  (x_i)^2 = \epsilon |x|^2,
\]
by means of the anticommutation relations~\eqref{Clifford} of $e_1,\dotsc,e_N$ and condition~\eqref{cond1}. 
For the classical case where $p_i$ is the $i$th partial derivative, the operator $\uD$ is the standard Dirac operator. 

The elements $e_1,\dotsc,e_N$ in fact generate what is known as a Clifford algebra~\cite{Porteous}, which we will denote as $\mathcal{C} = \mathcal{C} \ell(\mR^N)$. A general element in this algebra is a linear combination of products of $e_1,\dotsc,e_N$. The standard convention is to denote, for instance, $e_1e_2e_3$ simply as $e_{123}$. 
Hereto, we introduce the concept of a ``list'' for use as index of Clifford numbers. 
\begin{defi}\label{deflist}
We define a list to indicate a finite sequence of distinct elements of a given set, in our case the set $\{1,\dotsc,N\}$. 
For a list 
$A=a_1\dotsm a_n$ of $\{1,\dotsc,N\}$ with $0\leq n\leq N$, we will use the notation 
\begin{equation}\label{eA}
e_A = e_{a_1} e_{a_2} \dotsm e_{a_n}.
\end{equation}		
\end{defi}
\begin{rema}\label{remaeA}
Note that in a list the order matters as the Clifford generators $e_1,\dotsc,e_N$ anticommute. Moreover, duplicate elements would cancel out as they square to $\epsilon = \pm1$, so we consider only lists containing distinct elements. 
For a set $A = \{ a_1,\dotsc,a_n\} \subset \{1,\dotsc,N\}$, the notation $e_A$ stands for  $e_{a_1} e_{a_2} \dotsm e_{a_n}$ with $a_1 < a_2 < \dotsb < a_n$. 
\end{rema}

The collection $\{e_A\mid A\subset \{1,\dotsc,N\}\}$ forms a basis of the Clifford algebra $\mathcal{C}$, where for the empty set we put $e_{\emptyset}=1$. 

\begin{rema}
	In general, the square of each individual element $e_i$ $(i\in\{1,\dotsc,N\})$ can independently be chosen equal to either $+1$ or $-1$. 
	This corresponds to an underlying space with arbitrary signature defined by the specified signs. The original Dirac operator was constructed as a square root of the wave operator by means of the gamma or Dirac matrices which form a matrix realization of the Clifford algebra for $N=4$ with Minkowski signature. 
	
	To simplify notations in the following, we have chosen the square of all $e_i$ $(i\in\{1,\dotsc,N\})$ to be equal to $\epsilon$ which can be either $+1$ or $-1$. One can generalize all results to arbitrary signature by making the appropriate substitutions.
\end{rema}

In order to consider symmetries of the generalized Dirac operator \eqref{Dirac-op} we will work in the tensor product $\cA \otimes \mathcal{C}$ with the algebra $\cA$ as defined in Definition~\ref{defA}. To avoid overloading on notations, we will omit the tensor symbol $\otimes$ when writing down elements of $\cA \otimes \mathcal{C}$ and use regular product notation. In this notation, $e_1,\dotsc,e_N$ indeed commute with $x_i$ and $p_i$ for all $i\in\{1,\dotsc,N\}$. 

Akin to the realization of the Lie algebra $\mathfrak{sl}(2)$ in the algebra $\cA$ given by Theorem~\ref{theosl2}, we have something comparable in this case. 

\begin{theo}\label{theoosp}
	The algebra $\cA \otimes \mathcal{C}$ contains a copy of the Lie superalgebra $\mathfrak{osp}(1|2)$ generated by the (odd) elements $\uD$ and $\ux$ satisfying the relations
\begin{align*}
 \{\ux,\ux\} &= \epsilon \,2 |x|^2 & \{\uD,\uD\} &= \epsilon \,2 \Delta & \{\ux,\uD\} &= \epsilon \,2\mE\\
 \big[|x|^2,\ux\big] & = 0  & \big[|x|^2,\uD\big] & = -2\uD & [\mE,\ux] & = \ux \\
[\Delta,\ux] & = 2\ux& [\Delta,\uD] & = 0 & [\mE,\uD] & = -\uD 
\end{align*}	
and containing as an even subalgebra the Lie algebra $\mathfrak{sl}(2)$ in the algebra $\cA$ given by Theorem~\ref{theosl2} with relations
	\[
	\Big[ \mE , \frac{|x|^2}{2} \Big] = |x|^2 \qquad\quad  \Big[ \mE , -\frac{\Delta}{2} \Big] = -\Delta \qquad \Big[  \frac{|x|^2}{2} , -\frac{\Delta}{2} \Big] = \mE.
	\]
\end{theo}
\begin{proof}
	The relations follow by straightforward computations. 
By means of the anticommutation relations~\eqref{Clifford}, one finds that
\begin{align*}
\{\ux,\uD\} = \ & \sum_{i=1}^Nx_ie_i\sum_{j=1}^Np_je_j+\sum_{j=1}^Np_je_j\sum_{i=1}^Nx_ie_i \\*
= \ & \sum_{i=1}^N\epsilon(x_ip_i+p_ix_i)
+\sum_{1\leq i<j\leq N}(x_ip_j -p_jx_i-x_jp_i+p_ix_j  )e_ie_j.
\end{align*}
Looking back at \eqref{Euler}, the first summation is precisely $\epsilon \,2\mE$, while the second summation vanishes by relation~\eqref{cond2b}. Moreover, by relation~\eqref{cond2b} we have
\begin{align*}
[\mE,\uD] = \ & 
\frac12 \sum_{i=1}^N\sum_{j=1}^N[\{p_i,x_i\},p_j]e_j
= \  \frac12\sum_{i=1}^N\sum_{j=1}^N \{p_i,[x_i,p_j]\}e_j\\
= \ & \frac12\sum_{i=1}^N\sum_{j=1}^N \{p_i,[x_j,p_i]\}e_j = -\frac12\sum_{j=1}^N \bigg[\sum_{i=1}^N p_i^2,x_j\bigg] e_j = -\uD,
\end{align*}
and in the same manner, $[\mE,\ux]= \ux$. 
\end{proof} 

\subsection{Symmetries}\label{sec:3.2}

We wish to determine symmetries in the algebra  $\cA \otimes \mathcal{C}$ for the Dirac operator $\uD$ which are linear in both $x_1,\dotsc,x_N$ and $p_1,\dotsc,p_N$. Given the Lie superalgebra framework, it is natural to consider operators which either commute or anticommute with $\uD$. Indeed, the Lie superalgebra  $\mathfrak{osp}(1|2)$ has both a Scasimir and a Casimir element in its universal enveloping algebra~\cite{Casi}. The Scasimir operator
\begin{equation}
\label{Scasimir}
\mathcal{S} = \frac12\left( [\uD,\ux]-\epsilon\right) \in \mathcal{U}(\mathfrak{osp}(1|2))\subset \cA\otimes \mathcal{C},
\end{equation}
anticommutes with odd elements and commutes with even elements. In the classical case, the Scasimir operator is up to a constant term equal to the angular Dirac operator $\Gamma$, i.e.~$\uD$ restricted to the sphere. The Scasimir $\mathcal{S} $ is a symmetry which is linear in both $x_1,\dotsc,x_N$ and $p_1,\dotsc,x_N$ and we will get back to it before the end of this subsection. 
Finally, the square of the Scasimir element yields the Casimir element $\mathcal{C} = \mathcal{S}^2$, which commutes with all elements of $\mathfrak{osp}(1|2)$.

Note that another symmetry is obtained by means of the anticommutation relations~\eqref{Clifford} of the Clifford algebra. The so-called pseudo-scalar $e_1\dotsm e_N$ is easily seen to commute with $\uD$ for $N$ odd and anticommute with $\uD$ for $N$ even.

The Dirac operator is defined such that it squares to the Laplace operator, $\uD^2 = \epsilon \Delta$. This allows us to readily make use of the properties of $\Delta$  by means of the following straightforward relations. 
For an operator $Z$, we have that
\begin{equation}
\label{comacom}
[\uD,\{\uD,Z\}] = \uD(\uD Z + Z\uD ) - (\uD Z + Z\uD )\uD = \left[\uD^2,Z\right] 
\end{equation}
and
\begin{equation}
\label{acomcom}
\{\uD,[\uD,Z]\} = \uD(\uD Z - Z\uD ) + (\uD Z - Z\uD )\uD = \left[\uD^2,Z\right].
\end{equation}
A direct consequence of these relations is that every symmetry of the Laplace operator $\Delta$ yields symmetries of the Dirac operator $\uD$.

\begin{prop}	\label{SymDLaplace}	If $Z$ commutes with $\Delta$ and $\uD^2 = \epsilon \Delta$, then the operator $\{\uD,Z\}$ commutes with $\uD$, while the operator $[\uD,Z]$ anticommutes with $\uD$.
\end{prop}

Letting $Z$ be one of the symmetries of Theorem~\ref{theoSyms}, we indeed obtain symmetries of $\uD$, but they are not of the same order in $x_1,\dotsc,x_N$ as in $p_1,\dotsc,p_N$. These symmetries are in fact combinations of the obvious symmetries $p_1,\dotsc,p_N$ and symmetries which are linear in $x_1,\dotsc,x_N$ and in $p_1,\dotsc,p_N$. We set forth to determine the latter explicitly. Hereto, a first observation is that the elements of the Clifford algebra also commute with the Laplace operator, by definition as it is a scalar (non-Clifford) operator. For $A$ a list of distinct elements of $\{1,\dotsc,N\}$, we have
	\begin{align*}
	\uD ( \uD e_A \pm e_A \uD )  \mp   ( \uD e_A \pm e_A \uD )   \uD  
	=  \left[  	\uD^2 ,e_A  \right] =0.
	\end{align*}	
	The explicit expression of these operators follows from the anticommutation relations~\eqref{Clifford} as
	\begin{equation}\label{eAD}
	e_A \uD  =  e_A  \sum_{l=1}^N p_l e_l   =    (-1)^{|A|-1}\sum_{a\in A} p_a e_a e_A +   (-1)^{|A|}\sum_{a\notin A} p_a e_a e_A,
	\end{equation}	
with $|A|$ denoting the number of elements of the list $A$, so
	\begin{equation}\label{DeA}
\uD e_A - (-1)^{|A|} e_A \uD 	 = \sum_{a\in A} 2p_a e_a e_A
	\qquad\mbox{and}\qquad  
 \uD e_A + (-1)^{|A|} e_A \uD 	 = \sum_{a\notin A} 2p_a e_a e_A,
	\end{equation}
where (here and throughout the paper) the summation index $a\notin A$ is meant to run over all elements of $\{1,\dotsc,N\}\setminus A$. Note that for a list of one element $A=i$, we have $\uD e_i + e_i\uD = \epsilon 2p_i$. 
With this information, relations~\eqref{comacom} and \eqref{acomcom} now also lend themselves to the construction of more intricate symmetries of both $\uD$ and $\ux$.

\begin{theo}	\label{DOi}	In the algebra $\cA\otimes \mathcal{C}$, for $i \in \{1,\dotsc,N\}$, the operator
	\begin{equation}
\label{Oi}
	O_{i}  =   \frac{\epsilon}{2}\left( [\uD,x_i] -   e_i \right)   =   \frac{\epsilon}{2}\left( [p_i,\ux] -   e_i \right)    =   \frac{\epsilon}{2}\left(\sum_{l=1}^N  e_l C_{li}-   e_i \right)
	\end{equation}
	anticommutes with $\uD$ and $\ux$.
\end{theo}
\begin{proof}
The equalities in~\eqref{Oi} follow immediately from relation~\eqref{cond2b}, that is, $C_{ij}= [p_i,x_j]= [p_j,x_i]=C_{ji}$. 
By direct computation, using \eqref{acomcom} and the anticommutation relations~\eqref{Clifford}, we have
		\[
		\{   \uD, O_{i}\}  =  \frac{\epsilon}{2}\{   \uD, [\uD,x_i]\}	 -   \epsilon \frac12 \{   \uD, e_i\} 
		  =  \frac12 [\Delta,x_i]  -  p_i   = 0.
		\]
In the same manner, one finds that $O_{i}  =     \frac{\epsilon}{2}\left( [p_i,\ux] -   e_i \right) $ anticommutes with $\ux$.
	\end{proof}
The symmetries $O_i$ with one index $i \in \{1,\dotsc,N\}$ defined in \eqref{Oi} can be generalized to symmetries with multiple indices. 	
%
%
Hereto, we define the operators
\begin{equation}\label{DAxA}
\uD_A = \sum_{a\in A} p_a e_a \qquad\mbox{and}\qquad  \ux_A = \sum_{a\in A} x_a e_a,
\end{equation} 
for $A$ a subset of $\{1,\dotsc,N\}$, and by extension for $A$ a list of $\{1,\dotsc,N\}$ as the order does not matter in the summation. By means of the operators~\eqref{DAxA}, we state the following result. 
\begin{theo}\label{TOA}	In the algebra $\cA\otimes \mathcal{C}$, for $A$ a list of distinct elements of $\{1,\dotsc,N\}$, the operator
	\begin{align}
	\label{OA}
	O_{A}   = \ &  \frac{1}{2}\big(\uD\, \ux_A e_A - e_A \ux_A \uD - \epsilon e_A \big)\\* 
	\label{OAx}
	= \ &  \frac{1}{2}\big(e_A\uD_A\, \ux  -\ux \,\uD_A  e_A - \epsilon e_A \big)\\* 
	\label{OA1}
	= \ &  \frac12\bigg(-\epsilon+\sum_{j\in A}\sum_{i\notin A\setminus\{j\}}C_{ij}e_ie_j- \sum_{\{i, j\} \subset A} 2L_{ij} e_ie_j\bigg)e_A
	\end{align}	
satisfies 
		\[
		\uD \, O_{A}= (-1)^{|A|} O_{A} \uD \qquad\mbox{and}\qquad 	\ux \, O_{A}= (-1)^{|A|} O_{A} \ux .
		\]
\end{theo} 

\begin{proof} 
	We first show the equivalence of the three expressions~\eqref{OA} and \eqref{OA1}. Starting from \eqref{OA}, up to a factor $1/2$, and using $C_{ij}=C_{ji}$, we have
\begin{align*}
\ &  \uD\, \ux_A e_A - e_A \ux_A \uD - \epsilon e_A  
=    \sum_{l=1}^N p_l e_l\sum_{a\in A} x_a e_a e_A -e_A \sum_{a\in A} x_a e_a\sum_{l=1}^N p_l e_l  - \epsilon e_A\\
= \ &  \sum_{a\in A} \bigg(\epsilon(p_a  x_a - x_a  p_a) e_A 
+\sum_{l\in A\setminus\{a\}}( p_l x_a   + x_a  p_l   ) e_l e_ae_A
+ \sum_{l\notin A}(p_l x_a  - x_a  p_l )e_le_a e_A \bigg) -  \epsilon e_A\\
= \ &  \sum_{a\in A} \bigg(\epsilon(p_a  x_a - x_a  p_a) e_A 
+\sum_{l\in A\setminus\{a\}}( p_a x_l   + x_l  p_a   ) e_a e_le_A
+\sum_{l\notin A}(p_a x_l  - x_l  p_a )e_le_a e_A \bigg) - \epsilon e_A\\
=  \ &  e_A \sum_{a\in A} p_a e_a \sum_{l=1}^Nx_l e_l -  \sum_{l=1}^N x_l e_l\sum_{a\in A}p_a e_ae_A  - \epsilon e_A =  e_A\uD_A\, \ux  -\ux \,\uD_A  e_A - \epsilon e_A  .
\end{align*}	
Again starting from \eqref{OA}, up to a factor $1/2$, we have
\begin{align*}
	\ & \uD\, \ux_A e_A - e_A \ux_A \uD - \epsilon e_A 
	=    - \epsilon e_A+\sum_{l=1}^N p_l e_l\sum_{a\in A} x_a e_a e_A -e_A \sum_{a\in A} x_a e_a\sum_{l=1}^N p_l e_l  \\
	= \ &  \bigg(- \epsilon+ \sum_{l=1}^N p_l e_l\sum_{a\in A} x_a e_a  -\sum_{a\in A} x_a e_a\sum_{l\in A} p_l e_l + \sum_{a\in A} x_a e_a\sum_{l\notin A} p_l e_l   \bigg)e_A \\
  = \ &  \bigg(- \epsilon	
  + \epsilon \sum_{a\in A} (p_ax_a-x_ap_a) + 
  \sum_{a\in A}\sum_{l\in A\setminus\{a\}} ( p_l  x_a+x_a p_l) e_le_a +  \sum_{a\in A}\sum_{l\notin A}( p_l  x_a-x_a p_l) e_l e_a  \bigg)e_A \\
  = \ &  \bigg(- \epsilon	
   +  \sum_{a\in A}\sum_{l\notin A\setminus\{a\}}C_{la} e_l e_a 
    + 
  \sum_{\{a,l\}\subset A} (( p_l  x_a+x_a p_l) e_le_a+( p_a  x_l+x_l p_a) e_ae_l) \bigg)e_A 
\end{align*}	
which equals \eqref{OA1}, up to a factor $1/2$, when using $L_{ij}=x_ip_j-x_jp_i = p_jx_i-p_ix_j$  and $C_{ij}=C_{ji}$. 
	
Now for the proof itself, the case where $A$ is the empty set is trivial, as $O_{\emptyset}= -\epsilon/2$ obviously commutes with $\uD$ and $\ux$. 
For $A$ a singleton the result is given by Theorem~\ref{DOi}, so let now $|A|\geq 2$.	Using $\ux_A e_A = (-1)^{|A|-1} e_A \ux_A$ and \eqref{eAD}, we have
		\begin{align*}
		\uD \, O_{A}- (-1)^{|A|} O_{A} \uD  = \ 
		& \frac{1}{2}\uD \big(\uD\, \ux_A e_A - e_A \ux_A \uD - \epsilon e_A \big) - (-1)^{|A|}\frac{1}{2} \big(\uD\, \ux_A e_A - e_A \ux_A \uD - \epsilon e_A \big)\uD  \\
		 = \ 
		 & \frac{1}{2} \big(\uD^2\, \ux_A e_A  -  \ux_A e_A  \uD^2 \big)- \frac{\epsilon}{2}\big( \uD e_A  - (-1)^{|A|}e_A\uD  \big)   \\
		 = \ 
		 & \frac{\epsilon}{2} 	\sum_{a\in A} [\Delta,x_a]e_ae_A - 
	\epsilon	\sum_{a\in A}  p_ae_a e_{A} ,
		\end{align*}
		which vanishes because of condition \eqref{cond2}. In the same manner, using now the form \eqref{OAx} for $O_A$, the expression $\ux \, O_{A}- (-1)^{|A|} O_{A} \ux$ vanishes.
\end{proof}	
\begin{rema}
	Note that if the order of the list $A$ is altered, $O_A$ changes but only in sign. Say $\pi$ is a permutation on the list $A$, we have $O_{A} = \mathrm{sign}(\pi) O_{\pi(A)}$, where  $\mathrm{sign}(\pi)$ is positive for an even permutation $\pi$ and negative for an odd one. Hence, up to a sign all the symmetries of this form are given by $\{O_{A} \mid A \subset \{1,\dotsc,N\}\}$ where the elements of $A$ are in ascending order in accordance with the standard order for natural numbers.
\end{rema}

For the special case where $A =\{1,\dotsc,N\}$, the operator~\eqref{OA} is seen to correspond precisely to the Scasimir element~\eqref{Scasimir} of $\mathfrak{osp}(1|2)$ multiplied by the pseudo-scalar 
\begin{equation*}
O_{1\dotsm N} =\frac12\left( [\uD,\ux]-\epsilon\right) \, e_1\dotsm e_N.
\end{equation*}

For a list $A$ of $\{1,\dotsc,N\}$, the operator $O_A$ either commutes or anticommutes with $\uD$ and $\ux$. The subsequent corollary is useful if one is interested solely in commuting symmetries.
\begin{coro}\label{cor}
	In the algebra $\cA\otimes \mathcal{C}$, for $A$ a list of distinct elements of $\{1,\dotsc,N\}$, we have
	\[
	\Big[\uD,O_{A} \prod_{i\in A} O_i\Big]=0\,,\qquad 	\Big[\ux,O_{A} \prod_{i\in A} O_i\Big]=0 \rlap{\,.}
	\]
\end{coro}	
	Note that the order of the product matters in general, but not for the result.
	\begin{proof}
		Follows immediately from Theorem~\ref{DOi} and Theorem~\ref{TOA}.
		\end{proof}
		
Expression~\eqref{OA1} shows that the symmetries $O_A$ are constructed using the symmetries $C_{ij}$ and $L_{ij}$ from the previous section, together with the Clifford algebra generators $e_1,\dotsc,e_N$. The factor $1/2$ is chosen such that for a list of $\{1,\dotsc,N\}$ consisting of just two distinct elements $i$ and $j$, the symmetry $O_{ij}$ corresponds to the generalized angular momentum symmetry $L_{ij}$, up to additive terms. Indeed, we have that
\[
O_{ij} =L_{ij}-\epsilon\frac12 e_ie_j+\epsilon\frac12\sum_{l\neq j}C_{li} e_l e_j -\epsilon\frac12\sum_{l\neq i}C_{lj}  e_l e_i.
\]	
This can be written	more compactly by means of the explicit expression~\eqref{Oi} for $O_i$ as
\begin{equation}
	O_{ij} 	 =L_{ij}+\epsilon\frac12 e_ie_j+ O_i e_j - O_j  e_i  	\label{Oij}
\end{equation}	
Together with the Clifford algebra generators $e_1,\dotsc,e_N$, the symmetries $O_i$ with one index and $O_{ij}$ with two indices in fact suffice to build up all other symmetries $O_A$. Indeed, plugging in the expression~\eqref{Oi} for $O_i$, one easily verifies that
\begin{equation}	\label{OA2}
O_A	=   \bigg(\epsilon\frac{|A|-1}{2}+  \epsilon \sum_{i\in A} O_i e_i - \sum_{\{i, j\} \subset A} L_{ij} e_ie_j\bigg) e_A
\end{equation}
reduces to~\eqref{OA1}.  
Using now the form~\eqref{Oij} to substitute $L_{ij}$, the operator $O_A$ can also be written as
	\begin{equation}
	\label{OA3}
O_A	=  \Big(-\epsilon\frac{	(|A|-1)(|A|-2)}{4}- \epsilon (|A|-2)\sum_{i\in A} O_i e_i - \sum_{\{i, j\} \subset A} O_{ij} e_ie_j\Big) e_A.
	\end{equation}	
	
Finally, the symmetries can also be constructed recursively. If we denote by $A\setminus\{a\}$ the list $A$ with the element $a$ omitted, and we define 
$\mathrm{sign}(A,a)$ such that $\mathrm{sign}(A,a)e_{A\setminus\{a\}} e_a = e_A$, then it follows that
\begin{align*}
 \sum_{a\in A} \mathrm{sign}(A,a) O_{A\setminus\{a\}}e_a
 = \ 
 &   \sum_{a\in A} \Big(\epsilon\frac{|A|-2}{2}+ \epsilon \sum_{i\in A\setminus\{a\}} O_i e_i -  \sum_{\{i, j\} \subset A\setminus\{a\}} L_{ij} e_ie_j\Big) e_A
 \\
 = \ 
 &  \Big(\epsilon|A|\frac{|A|-2}{2}+ \epsilon (|A|-1)\sum_{i\in A} O_i e_i - (|A|-2) \sum_{\{i, j\} \subset A} L_{ij} e_ie_j\Big) e_A \\
 = \ 
 & \epsilon\frac{|A|-2}{2}e_A +\epsilon  \sum_{i\in A} O_i e_i  e_A + (|A|-2) O_A.
	\end{align*}
Using this relation, Theorem~\ref{TOA} can also be proved by induction on the cardinality of $A$, starting from $|A|=3$.
%

\subsection{Symmetry algebra}
\label{ssec:3.4}

Before establishing the algebraic structure generated by the symmetries $O_A$, we first introduce some helpful relations with Clifford numbers. From the definition~\eqref{Oi} of $O_j$, we have
	\[
	\{e_i,O_j\}= e_iO_j +O_je_i = [p_i,x_j] - \delta_{ij}.
	\]
The property $[p_i,x_j]=[p_j,x_i]$ then implies that $\{e_i,O_j\}= \{e_j,O_i\}$, or by a reordering of terms $O_ie_j-O_je_i = e_iO_j-e_jO_i$. This is in fact a special case of the following useful result.
\begin{lemm}\label{lemmA}
	In the algebra $\cA\otimes \mathcal{C}$, for $A$ a list of distinct elements of $\{1,\dotsc,N\}$, we have
	\[
	\sum_{a\in A}  O_a    e_a e_A =  e_A\sum_{a\in A}e_aO_a  .   
	\]
\end{lemm}	
\begin{proof}
	The identity follows by direct calculation using the definition~\eqref{Oi} of $O_a$ and the commutation relations of $e_1,\dotsc,e_N$: 
	\begin{align*}	
		\sum_{a\in A}  O_a    e_a e_A 
		= \ & \sum_{a\in A}  \epsilon\frac12\sum_{l=1}^N  [p_l,x_a]e_le_a e_A 
		- \sum_{a\in A}  \epsilon\frac12 e_a    e_a e_A \\
		= \ & \sum_{a\in A}  \epsilon\frac12\sum_{l\in A}  [p_l,x_a]e_Ae_le_a 
		-\sum_{a\in A}  \epsilon\frac12\sum_{l\notin A}  [p_l,x_a]e_Ae_le_a 
		- \sum_{a\in A}  \frac12 e_A  \\
	= \ & \sum_{a\in A}  \epsilon\frac12\sum_{l\in A}  [p_a,x_l]e_Ae_ae_l 
	+\sum_{a\in A}  \epsilon\frac12\sum_{l\notin A}  [p_l,x_a]e_Ae_ae_l  
	-\sum_{a\in A}  \epsilon\frac12 e_Ae_a    e_a 
	 	=   e_A\sum_{a\in A}e_aO_a.
	\end{align*}
\end{proof}
Note that by means of this lemma, the symmetry $O_A$, in the form~\eqref{OA2}, can equivalently be written with $e_A$ in front, that is 
	\begin{equation*}
	O_{A}   = e_A\Big(\epsilon\frac{|A|-1}{2}+  \epsilon \sum_{i\in A} e_iO_i  - \sum_{\{i, j\} \subset A} L_{ij} e_ie_j\Big) .
	\end{equation*}

\subsubsection{Relations for symmetries with one or two indices}	

Next, we present some relations which hold for symmetries with one or two indices.

\begin{theo}\label{theoOijOk}
	In the algebra $\cA\otimes \mathcal{C}$, for $i,j,k\in\{1,\dotsc,N\}$ we have
\begin{equation*}
	[O_{ij},O_k] + [O_{jk},O_i]+[O_{ki},O_j] = 0.
\end{equation*}
\end{theo}
\begin{proof}
If any of the indices are equal, the identity becomes trivial as $O_{ij}=-O_{ji}$.	For distinct $i,j,k$, we have by \eqref{Oi}, \eqref{Oij} and using $O_ie_j-O_je_i = e_iO_j-e_jO_i$
\begin{align*}
	[O_{ij},O_k] + [O_{jk},O_i]+[O_{ki},O_j]  =   \ 	&  \epsilon\frac12\sum_{l=1}^N  e_l(	[L_{ij} ,C_{lk}]+[L_{jk},C_{li}] +[L_{ki},C_{lj}])\\*
	 &  + \epsilon \frac12 ([e_ie_j,O_k] + [e_je_k,O_i]+[ e_ke_i,O_j])  \\* 
	&  +    (O_i e_j - O_j  e_i)O_k-O_k(e_i O_j - e_j  O_i)  +  (O_j e_k - O_k  e_j)O_i\\* 
	&  -O_i (e_j O_k - e_k  O_j) +  (O_k e_i - O_i  e_k)O_j-O_j(e_k O_i - e_i  O_k) .
\end{align*}
The first line of the right-hand side vanishes by Theorem~\ref{theoLAlg}, while the second line does so by direct calculation plugging in the definition of~$O_i$ and using $C_{ij}=C_{ji}$. Finally, the remaining terms of the last two lines cancel out pairwise.
\end{proof}

For the next result, we first write out the form~\eqref{OA2} of $O_A$ for $A$ a list of three elements, say $i,j,k$ which are all distinct:
	\begin{equation}
	O_{ijk} 	 = \epsilon e_ie_je_k+   O_i e_je_k-O_j e_ie_k+O_k e_ie_j   +L_{ij} e_k -L_{ik} e_j +L_{jk} e_i.	\label{Oijk}
	\end{equation}		
The commutation relations for symmetries with two indices are as follows.
\begin{theo}\label{theoAlgOii}
	In the algebra $\cA\otimes \mathcal{C}$, for  $i,j,k,l\in\{1,\dotsc,N\}$ the symmetries 
satisfy
\begin{align*}
	[O_{ij},O_{kl}]  
	=\ & (O_{il}+ \epsilon [O_i,O_l])\delta_{jk} +(O_{jk}+ \epsilon [O_j,O_k])\delta_{il}+(O_{ki}+ \epsilon [O_k,O_i] )\delta_{lj}+(O_{lj}+ \epsilon [O_l,O_j])\delta_{ki} 
	\\ 
	&+\frac12( \{O_i,O_{jkl}\}
	-\{O_j,O_{ikl}\}-\{O_{ijl},O_k\} 
	+\{O_{ijk},O_l\}
	).
\end{align*}
\end{theo}
\begin{proof}
For the cases where $i=j$ or $k=l$ or $\{i,j\}=\{k,l\}$, both sides of the equation reduce to zero, so from now on we assume that $i\neq j$ and $k\neq l$ and $\{i,j\}\neq\{k,l\}$. 
Plugging in \eqref{Oij}, we have
\begin{align*}
	[O_{ij},O_{kl}]  =   \ &\left[ L_{ij}+\epsilon\frac12 e_ie_j+ O_i e_j - O_j  e_i 
	,L_{kl} +\epsilon\frac12 e_ke_l+ O_k e_l - O_l  e_k \right]   \\
	=\ & [L_{ij},L_{kl}] 
	+ \left[ L_{ij} , O_k e_l - O_l  e_k \right]  
	+ \left[ O_i e_j - O_j  e_i ,L_{kl}  \right]+[  O_i e_j - O_je_i , O_k e_l - O_le_k ] \\*
	&+\frac14\left[  e_ie_j, e_ke_l \right]+\epsilon\frac12\left[  e_ie_j ,  O_k e_l \right] 
	+\epsilon\frac12\left[ e_ie_j ,   - O_l e_k \right] 
	+\epsilon\frac12\left[  O_i e_j ,  e_ke_l \right]
	+\epsilon\frac12\left[   - O_j  e_i ,  e_ke_l \right]	
\end{align*}
By Theorem~\ref{theoLAlg}, and using $\{e_i,O_j\} = [p_i,x_j] - \delta_{ij}$, we have 
\begin{align*}
[ L_{ij}, L_{kl}] = \ &  L_{il}\delta_{jk} +L_{jk}\delta_{il}+L_{ki}\delta_{lj}+L_{lj}\delta_{ki}\\*
 & + L_{il}\{e_j,O_k\} +L_{jk}\{e_i,O_l\}+L_{ki}\{e_l,O_j\}+L_{lj}\{e_k,O_i\}.
\end{align*}
Using $\{e_i,O_j\} = \{e_j,O_i\}=\frac12\{e_i,O_j\}+ \frac12\{e_j,O_i\} $, the terms in the last line can be rewritten as
\begin{align*}
 & \frac12( (L_{il}e_j+L_{lj}e_i)O_k +(L_{jk}e_i+L_{ki}e_j)O_l+(L_{lj}e_k+L_{jk}e_l)O_i+(L_{ki}e_l+L_{il}e_k)O_j) \\
 & +\frac12( (L_{jk}O_l+L_{lj}O_k)e_i +(L_{il}O_k+L_{ki}O_l)e_j+(L_{lj}O_i+L_{il}O_j)e_k+(L_{ki}O_j+L_{jk}O_i)e_l) .
\end{align*}
Together with 
\begin{align*}
[ 	L_{ij} , O_k e_l - O_l  e_k ]   =\  & 
 	[L_{ij} , O_k ]e_l -[L_{ij} , O_l]  e_k =   \{L_{ij}e_k,O_l\} - \{L_{ij}e_l,O_k\}
 \\
[ 	 O_i e_j - O_j  e_i ,	L_{kl}  ]   =\ & 
[O_i , L_{kl}]  e_j -[ O_j, L_{kl}]   e_i
   =  \{L_{kl}e_j,O_i\} - \{L_{kl}e_i,O_j\}
,
\end{align*}
we find that $[ L_{ij}, L_{kl}] + [ 	L_{ij} , O_k e_l - O_l  e_k ]  + [ 	 O_i e_j - O_j  e_i ,	L_{kl}  ]$ equals
\begin{align*}
&	 L_{il}\delta_{jk} +L_{jk}\delta_{il}+L_{ki}\delta_{lj}+L_{lj}\delta_{ki}    
\\* 
&+\frac12( (L_{il}e_j+L_{lj}e_i)O_k +  \{- L_{ij}e_l,O_k\} +(L_{jk}e_i+L_{ki}e_j)O_l +  \{  L_{ij}e_k,O_l\}\\* 
&
+(L_{lj}e_k+L_{jk}e_l)O_i+   \{   O_i,  L_{kl}e_j\}
+(L_{ki}e_l+L_{il}e_k)O_j + \{  O_j , - L_{kl}e_i\})  \\*
& +\frac12( (L_{jk}O_l+L_{lj}O_k+ [	L_{kl} , O_j]  )e_i 
+(L_{il}O_k+L_{ki}O_l- [	L_{kl} , O_i])e_j \\*
&+(L_{lj}O_i+L_{il}O_j-[	L_{ij}, O_l ]  )e_k 
+(L_{ki}O_j+L_{jk}O_i+[	L_{ij}, O_k ])e_l  ) .
\end{align*}
This simplifies to
\begin{align*}
&	 L_{il}\delta_{jk} +L_{jk}\delta_{il}+L_{ki}\delta_{lj}+L_{lj}\delta_{ki}  +\frac12( \{L_{il}e_j+L_{lj}e_i+L_{ji}e_l,O_k\} \\* 
& 
+\{L_{jk}e_i+L_{ki}e_j+L_{ij}e_k,O_l\}+\{O_i,L_{lj}e_k+L_{jk}e_l+L_{kl}e_j\}+\{O_j,L_{ki}e_l+L_{il}e_k+L_{lk}e_i\}
) ,
\end{align*}
by means of 
\begin{align*}
&  ([L_{jk},O_l]+[L_{lj},O_k]+ 	[L_{kl} , O_j ] )e_i 
+([L_{il},O_k]+[L_{ki},O_l]- [	L_{kl} , O_i])e_j \\*
&+([L_{lj},O_i]+[L_{il},O_j]-[	L_{ij}, O_l]   )e_k 
+([L_{ki},O_j]+[L_{jk},O_i]+[	L_{ij}, O_k] )e_l  
=0 ,
\end{align*}
which is a direct consequence of Theorem~\ref{theoLAlg} after plugging in the definition~\eqref{Oi} of $O_i$.

Now, the other terms appearing in $[O_{ij},O_{kl}]$ can be expanded as follows. First,  \[
	[  e_ie_j, e_ke_l ]=2\epsilon( \delta_{jk} e_ie_l -\delta_{il}e_ke_j-\delta_{lj}e_ie_k+\delta_{ik}e_le_j ). 
\]	
Moreover, we have
\begin{align*}
		[  e_i e_j , O_k e_l]= \ & e_i e_jO_k e_l - O_k e_le_i e_j \\
		= \ & e_i e_j \{e_l,O_k\}-e_i e_j e_lO_k + O_k e_ie_l e_j - O_k (\epsilon2\delta_{il}) e_j \\
		= \ & e_i e_j \{e_l,O_k\}-e_i e_j e_lO_k  - O_k e_i e_je_l+ O_k e_i(\epsilon2\delta_{jl}) - O_k (\epsilon2\delta_{il}) e_j
		.
\end{align*} 
As $\{e_k,O_l\}= \{e_l,O_k\}$, after interchanging $k$ and $l$ in this result, subtraction yields the following
\[ 
[  e_i e_j , O_k e_l-O_le_k]= 
\delta_{jl} \epsilon 2O_ke_i- \delta_{il} \epsilon 2O_ke_j 
-\delta_{jk} \epsilon 2O_le_i+ \delta_{ik} \epsilon 2O_le_j  +\{e_ie_je_k,O_l\}  -\{e_ie_je_l,O_k\}.
\]

For the last term, $[  O_i e_j - O_je_i , O_k e_l - O_le_k ]$, we use Lemma~\ref{lemmA} to find
\begin{align*}
		2	[  O_i e_j - O_je_i , O_k e_l - O_le_k ]= \ &  
		2(O_i e_j - O_je_i )(e_kO_l - e_l O_k  )
		- 2( O_k e_l - O_le_k )(e_iO_j  - e_jO_i ) \\
		= \ &  (  O_k e_le_j - O_le_ke_j )O_i+ O_i (e_je_k  O_l  - e_je_l O_k  ) \\*
		& + ( - O_k e_le_i + O_le_ke_i )O_j+O_j(- e_i e_kO_l +e_i e_l O_k  )  \\*
		& +  (-O_i e_je_l + O_je_i e_l )O_k +O_k(- e_le_iO_j  + e_le_jO_i  )\\*
		&+ (  O_i e_j e_k- O_je_i e_k )  O_l+ O_l(e_k e_iO_j  -e_k e_jO_i   ).
\end{align*}
We first consider the case where one of $i,j$ is equal to either $k$ or $l$, for instance, say $i=l$:
\begin{align*}
	 &  (  O_k e_ie_j - O_ie_ke_j )O_i+ O_i (e_je_k  O_i  - e_je_i O_k  )  + ( - O_k \epsilon + O_ie_ke_i )O_j+O_j(- e_i e_kO_i +\epsilon O_k  )  \\*
		& +  (-O_i e_je_i + O_j\epsilon )O_k +O_k(- \epsilon O_j  + e_ie_jO_i  ) +(  O_i e_j e_k- O_je_i e_k )  O_i+ O_i(e_k e_iO_j  -e_k e_jO_i   ) \\	
		= \ &  (  O_k e_ie_j - O_ie_ke_j-O_j e_i e_k+O_i e_j e_k- O_je_i e_k+O_k  e_ie_j  )O_i\\*		
		&  
		+ O_i (e_je_k  O_i  - e_je_i O_k +   e_ke_i O_j+e_k e_iO_j  -e_k e_jO_i  - e_je_i  O_k  )    +\epsilon 2[O_j,O_k] 
\end{align*}	
Using Lemma~\ref{lemmA}, this equals	
\begin{align*}	
	 \ &  (  O_k e_ie_j - O_ie_ke_j-O_j e_i e_k+  O_i e_j e_k- O_je_i e_k+O_k  e_ie_j )O_i\\*
	&+ O_i (  O_ie_je_k  -  O_ke_je_i +   O_je_ke_i + O_je_k e_i  -O_ie_k e_j  -   O_ke_je_i )  +\epsilon 2[O_j,O_k]     	 \\
	= \ &  2\{  O_k e_ie_j -O_j e_i e_k+  O_i e_j e_k ,O_i\}  +\epsilon 2[O_j,O_k]     		.
\end{align*}		
Finally, when $i,j,k,l$ are all distinct, $	2	[  O_i e_j - O_je_i , O_k e_l - O_le_k ]$ equals
\begin{align*}
		 & (  O_k e_le_j - O_le_ke_j )O_i+ O_i (-e_ke_j  O_l  + e_le_j O_k +e_ke_l  O_j -e_ke_l  O_j)  
		\\*
		& +( - O_k e_le_i + O_le_ke_i )O_j +  O_j( e_ke_i O_l-e_le_i  O_k +e_le_k O_i-e_le_k O_i )
		\\*
		&  +(-O_i e_je_l + O_je_i e_l )O_k+ O_k( e_ie_lO_j  - e_je_lO_i +e_je_iO_l  - e_je_iO_l ) 
		\\
		&  +(  O_i e_j e_k- O_je_i e_k )  O_l+ O_l(-e_ie_k O_j  +e_je_k O_i  +e_ie_j O_k-e_ie_j O_k ) \\
		= \ &  (  O_k e_le_j - O_le_ke_j -O_je_le_k )O_i + O_i (-O_l e_ke_j   +  O_ke_le_j + O_je_ke_l  )
		\\*
		& +( - O_k e_le_i + O_le_ke_i - O_ie_ke_l )O_j + O_j(  O_le_ke_i-O_ke_le_i   + O_ie_le_k ) 
		\\*
		&  +( -O_i e_je_l + O_je_i e_l -O_le_ie_j)O_k + O_k( O_je_ie_l  -O_i e_je_l +O_le_je_i  )
		\\*
		& + (  O_i e_j e_k- O_je_i e_k - O_ke_je_i )  O_l + O_l(-O_je_ie_k   +O_ie_je_k   +O_ke_ie_j  )\\
		= \ &  \{  O_k e_le_j - O_le_ke_j + O_je_ke_l  ,O_i\}
		+\{  O_le_ke_i - O_k e_le_i + O_ie_le_k  ,O_j\} 
		\\*
		& + \{   O_je_i e_l -O_i e_je_l+O_le_je_i ,O_k \}
		+ \{  O_i e_j e_k- O_je_i e_k +O_ke_ie_j,O_l \}  .
\end{align*}
This boils down to $[  O_i e_j - O_je_i , O_k e_l - O_le_k ]$ being equal to
\begin{align*}
& 	 \epsilon\delta_{jk}[O_i , O_l]  + \epsilon\delta_{il}[O_j, O_k]+\epsilon\delta_{lj}[O_k , O_i]  +  \epsilon\delta_{ki}[O_l,O_j] \\*
	& + \frac12(\{  O_k e_le_j - O_le_ke_j + O_je_ke_l  ,O_i\}
	+\{  O_le_ke_i - O_k e_le_i + O_ie_le_k  ,O_j\} 
	\\*
	& + \{   O_je_i e_l -O_i e_je_l+O_le_je_i ,O_k \}
	+ \{  O_i e_j e_k- O_je_i e_k +O_ke_ie_j,O_l \}  ).
\end{align*}

Hence, combining all of the above, we find
	\begin{align*}
[O_{ij},O_{kl}]	=\ & L_{il}\delta_{jk} +L_{jk}\delta_{il}+L_{ki}\delta_{lj}+L_{lj}\delta_{ki}  +\frac12( \{L_{il}e_j+L_{lj}e_i+L_{ji}e_l,O_k\} 
	\\*
	& +\{L_{jk}e_i+L_{ki}e_j+L_{ij}e_k,O_l\}
	+\{L_{ki}e_l+L_{il}e_k+L_{lk}e_i,O_j\}\\*
	&
	+\{L_{lj}e_k+L_{jk}e_l+L_{kl}e_j,O_i\})  +\epsilon\frac12( \delta_{jk} e_ie_l -\delta_{il}e_ke_j-\delta_{lj}e_ie_k+\delta_{ik}e_le_j ) \\*
 & +\epsilon\delta_{jk}[O_i , O_l]  + \epsilon\delta_{il}[O_j, O_k]+\epsilon\delta_{lj}[O_k , O_i]  +  \epsilon\delta_{ki}[O_l,O_j] \\*
 &  +\frac12(\{  O_k e_le_j - O_le_ke_j + O_je_ke_l  ,O_i\}
 +\{  O_le_ke_i - O_k e_le_i + O_ie_le_k  ,O_j\} 
 \\*
 & + \{   O_je_i e_l -O_i e_je_l+O_le_je_i ,O_k \}
 + \{  O_i e_j e_k- O_je_i e_k +O_ke_ie_j,O_l \}  )\\*
	&+\epsilon\frac12( \delta_{jl} \epsilon 2O_ke_i - \delta_{il} \epsilon 2O_ke_j  - \{e_ie_je_l,O_k\}  -\delta_{jk} \epsilon 2O_le_i + \delta_{ik} \epsilon 2O_le_j + \{e_ie_je_k,O_l\} ) \\*
	&
	-\epsilon\frac12( \delta_{lj} \epsilon 2O_ie_k - \delta_{kj} \epsilon 2O_ie_l - \{e_ke_le_j,O_i\}  -   \delta_{li} \epsilon 2O_je_k+ \delta_{ki} \epsilon 2O_je_l + \{e_ke_le_i,O_j\} ).
		\end{align*}
Collecting the appropriate terms, we recognize all ingredients to make symmetries with three indices \eqref{Oijk} and we arrive at the desired result	
	\begin{align*}
	[O_{ij},O_{kl}]	=\ & (O_{il}+ \epsilon [O_i,O_l])\delta_{jk} +(O_{jk}+ \epsilon [O_j,O_k])\delta_{il}+(O_{ki}+ \epsilon [O_k,O_i] )\delta_{lj}+(O_{lj}+ \epsilon [O_l,O_j])\delta_{ki} 
	\\*
	&+\frac12( \{O_i,O_{jkl}\}
	-\{O_j,O_{ikl}\} -\{O_{ijl},O_k\} 
	+\{O_{ijk},O_l\}
	).\qedhere
	\end{align*}
\end{proof}

In summary, for $i,j,k,l$ all distinct elements of $\{1,\dotsc,N\}$ we have
	\begin{align*}
	[O_{ij},O_{ki}]  
	=\ & O_{jk}+ \epsilon [O_j,O_k] 
	+\{O_{ijk},O_i\}
	\\ 
	[O_{ij},O_{kl}]  
	= \ 	&\frac12( \{O_i,O_{jkl}\}
	-\{O_j,O_{ikl}\}-\{O_{ijl},O_k\} 
	+\{O_{ijk},O_l\}
	).
	\end{align*}

\subsubsection{Relations for symmetries with general index}	
	
Now, we are interested in relations for general symmetries with an arbitrary index list $A$, i.e.
\[
	O_{A}  =  \frac{1}{2}\big(\uD\, \ux_{A} e_A -e_A  \ux_{A}  \uD - \epsilon e_A \big).
\]
Before doing so, we make a slight detour clearing up some conventions and notations. 
An important fact to take into account is the interaction of the appearing Clifford numbers with different lists as index (recall Definition~\ref{deflist}). For instance, if $A$ and $B$ denote two lists of $\{1,\dotsc,N\}$, the following properties are readily shown to hold by direct computation:
\begin{equation}\label{eAeB}
e_Be_A =  (-1)^{|A||B|-|A\cap B|}e_Ae_B,
\end{equation}
and
\begin{equation}\label{eA2}
(e_A)^2 = \epsilon^{|A|}(-1)^{\frac{|A|^2-|A|}{2}}.
\end{equation}
The product $e_Ae_B$ in fact reduces (up to a sign) to contain only $e_i$ with $i$ an index in the set ${(A\cup B)\setminus(A\cap B)}$, since all indices in $A\cap B$ appear twice and cancel out (for these set operations we disregard the order of the lists $A$ and $B$ and view them just as sets). The remaining indices are what is called the symmetric difference of the sets $A$ and $B$. We will denote this associative operation by $A\triangle B = (A\cup B)\setminus(A\cap B) = (A\setminus B)\cup(B\setminus A)$. 
When applied to two lists $A$ and $B$, we view them as sets and the resultant $A\triangle B$ is a set. 
Note that $|A\triangle B| = |A| +|B|-2|A\cap B|$.  

When dealing with interactions between symmetries $O_A$ and $O_B$ for two lists $A$ and $B$, products of the kind $e_Ae_B$ are exactly what we encounter. 
Not wanting to overburden notations, but still taking into account all resulting signs due to the anticommutation relations~\eqref{Clifford} if one were to work out the reduction $e_Ae_B$, we introduce the following notation
\begin{equation}\label{OAB}
O_{A, B}  =  \frac{1}{2}\big(\uD\, \ux_{A\triangle B} e_Ae_B -e_Ae_B  \ux_{A\triangle B}  \uD - \epsilon e_Ae_B \big).
\end{equation}	
Note that the order of $A$ and $B$ matters, as by~\eqref{eAeB} we have $O_{A, B} = (-1)^{|A||B|-|A\cap B|}O_{B, A}$. Moreover, up to a sign $O_{A, B}$ is equal to $O_{A\triangle B}$, where the elements of $A\triangle B$ are in ascending order when used as a list, see Remark~\ref{remaeA}. 
Since the symmetric difference operation on sets is associative, one easily extends this definition to an arbitrary number of lists, e.g.
\[
O_{A, B,C}  =  \frac{1}{2}\big(\uD\, \ux_{A\triangle B\triangle C} e_Ae_Be_C -e_Ae_Be_C  \ux_{A\triangle B\triangle C}  \uD - \epsilon e_Ae_Be_C \big).
\]	
As an example, if we consider the lists $A=234$ and $B=31$, then the set $A\triangle B$ contains the elements $2,4,1$ but not $3$ as $3$ appears in both $A$ and $B$, so we have
\begin{align*}
O_{234,31}  = \ & \frac{1}{2}\big(\uD\, \ux_{\{2,4,1\}} e_{234}e_{31} -e_{234}e_{31}  \ux_{\{2,4,1\}}  \uD - \epsilon e_{234}e_{31} \big)\\
= \ & - \epsilon \frac{1}{2}\big(\uD\, (x_2e_2 + x_4e_4 + x_1e_1) e_{241} -e_{241}  (x_2e_2 + x_4e_4 + x_1e_1)  \uD - \epsilon e_{241} \big) \\
= \ & - \epsilon O_{241} = - \epsilon O_{124}.
\end{align*}

We elaborate on one final convention. If $A$ and $B$ denote two lists of $\{1,\dotsc,N\}$, then when viewed as ordinary sets, the intersection $A\cap B$ contains all elements of $A$ that also belong to $B$ (or equivalently, all elements of $B$ that also belong to $A$). As a list appearing as index of a Clifford number, we distinguish between $A\cap B$ and $B\cap A$ in the sense that we understand the elements of $A\cap B$ to be in the sequential order of $A$, while those of $B\cap A$ are in the sequential order of $B$. If $A=124$ and $B=231$, then $e_{A\cap B} = e_{12}$, while $e_{B\cap A} = e_{21}$.

The framework where operators of the form~\eqref{OAB} make their appearance is one where both commutators and anticommutators are considered. Inspired by property~\eqref{eAeB}, we define the ``supercommutators''
\begin{align}\label{scom}
	\llbracket O_{A},O_{B}\rrbracket_- =\ &  O_AO_B - (-1)^{|A||B|-|A\cap B|} O_BO_A,
\\ \label{sacom}
	\llbracket O_{A},O_{B}\rrbracket_+ = \ & O_AO_B + (-1)^{|A||B|-|A\cap B|} O_BO_A.
\end{align}
The algebraic relations we obtained in Theorem~\ref{theoAlgOii} can now be generalized to higher index versions. We start with a generalization of Theorem~\ref{theoOijOk}. 
\begin{theo}\label{OAa}
	In the algebra $\cA\otimes \mathcal{C}$, for $A$ a list of distinct elements of $\{1,\dotsc,N\}$, we have
\begin{equation*}
	\sum_{a\in A} \llbracket O_{a},O_{a, A}\rrbracket_-  = 0  . 
\end{equation*}
\end{theo} 
Note that $O_{a, A}$ is up to a sign equal to $O_{ A\setminus\{a\}  }$ since $a\in A$, where $A\setminus\{a\}$ is the list $A$ with the element $a$ removed.
\begin{proof}
By direct calculation and using \eqref{Oi}, \eqref{OA2} and Lemma~\ref{lemmA} to arrive at the second line we have
\begin{align*}
		& 	\sum_{a\in A} (O_{a}O_{a, A}  -(-1)^{|A|-1}  O_{a, A} O_{a}) 	
	=   \     \epsilon \frac{|A|-2}{2} \Big(\sum_{a\in A}O_ae_ae_A-(-1)^{|A|-1} \sum_{a\in A} e_ae_A O_{a} \Big)  \\* 
	&  \qquad \qquad \qquad \qquad \qquad \qquad\qquad+ \epsilon\sum_{a\in A}   \sum_{b\in A\setminus\{a\}}  O_{a} e_ae_{A}e_bO_b   -(-1)^{|A|-1} \epsilon\sum_{a\in A}   \sum_{b\in A\setminus\{a\}} O_b e_b e_a e_{A} O_a	\\*
		&  +\epsilon\frac12	\sum_{a\in A} 	\sum_{\{i,j\}\subset A\setminus\{a\}}\sum_{l\notin A\setminus\{a,i,j\}}  	[L_{ij} ,C_{la}]e_le_ie_je_ae_A   +  \epsilon\frac12	\sum_{a\in A} 	\sum_{\{i,j\}\subset A\setminus\{a\}}\sum_{l\in A\setminus\{a,i,j\}}  	\{L_{ij} ,C_{la}\}e_le_ie_je_ae_A \\* 
	&  -  \epsilon\frac12	\sum_{a\in A} 	\sum_{\{i,j\}\subset A\setminus\{a\}} 	L_{ij} (e_a e_ie_je_ae_A - 	(-1)^{|A|-1}  e_ie_je_ae_A e_a).
\end{align*}	
As $(-1)^{|A|-1}e_ae_A = e_Ae_a$ for $a\in A$, the last line vanishes identically. Similarly, the first line of the right-hand side vanishes by Lemma~\ref{lemmA}. Moreover, the second line then vanishes by interchanging the summations.  In the third line, we write the first two summations as one summation running over all three-element subsets of $A$ as follows
\[
	\epsilon\frac12		\sum_{\{a,i,j\}\subset A}\sum_{l\notin A\setminus\{a,i,j\}}   	([L_{ij} ,C_{la}]+[L_{ai} ,C_{lj}]+[L_{ja} ,C_{li}])e_le_ie_je_ae_A  	,
\]
which vanishes by Theorem~\ref{theoLAlg}. Finally, rewriting the four summations in the fourth line as one summation over all four-element subsets of $A$, one sees that all terms in this summation cancel out using $C_{ij}=C_{ji}$ and $L_{ij}=-L_{ji}$.

\end{proof}

Using Theorem~\ref{OAa} for the list $A=ijkl$ yields
\[
\{O_i,O_{jkl}\}
-\{O_j,O_{ikl}\}
+\{O_k,O_{ijl}\} 
-\{O_l,O_{ijk}\}
 =0.
\]
By means of this identity, the relation of Theorem~\ref{theoAlgOii} for four distinct indices can be cast also in two other formats:
\begin{equation}\label{Oijkl}
	\llbracket O_{ij},O_{kl}\rrbracket_-=  [O_{ij},O_{kl}]  
=  \{O_i,O_{jkl}\}
-\{O_j,O_{ikl}\}
= -\{O_{ijl},O_k\} 
+\{O_{ijk},O_l\}
.
\end{equation}

The results of Theorem~\ref{theoAlgOii} are actually special cases of two different more general relations. By means of the following three theorems and the supercommutators~\eqref{scom}--\eqref{sacom}, one is able to swap two symmetry operators $O_A$ and $O_B$ for $A,B$ arbitrary lists of $\{1,\dotsc,N\}$. Moreover, the supercommutators reduce to explicit expressions in terms of the symmetries and supercommutators containing (at least) one symmetry with just one index. 
First, we generalize ~\eqref{Oijkl} to arbitrary disjoint lists. 
\begin{theo}\label{theoAdisjB}
		In the algebra $\cA\otimes \mathcal{C}$, for two lists of $\{1,\dotsc,N\}$, denoted by $A$ and $B$, such that $A\cap B=\emptyset$ as sets, we have
\[	
	\llbracket O_{A},O_{B}\rrbracket_-  = 
	  \epsilon \sum_{a\in A} \llbracket O_a, O_{a, A, B}\rrbracket_-  \,  .
\]
\end{theo}
Note that in this case ${a\triangle A\triangle B}= {(A\setminus \{a\})\cup B  }$. Moreover, following a similar strategy (or using Theorem~\ref{OAa}) one also obtains 
\[	
\llbracket O_{A},O_{B}\rrbracket_-  = 
\epsilon \sum_{b\in B} \llbracket  O_{ A, B, b},O_b\rrbracket_-\, .
\] 
\begin{proof}
A practical property for this proof and the following ones is \eqref{DeA}. 	
By definition, plugging in \eqref{OA}, we find
\begin{align*}
	\llbracket O_{A},O_{B}\rrbracket_-  = \ & O_AO_B - (-1)^{|A||B|-|A\cap B|} O_BO_A \\
	= \ & \frac14\big(\uD\, \ux_A e_A - e_A\ux_A   \uD - \epsilon e_A \big)\big(\uD\, \ux_B e_B - e_B \ux_B  \uD - \epsilon e_B \big) \\
	&  - (-1)^{|A||B|} \frac14 \big(\uD\, \ux_B e_B - e_B \ux_B  \uD - \epsilon e_B \big)\big(\uD\, \ux_A e_A -  e_A\ux_A  \uD - \epsilon e_A \big).
\end{align*}
First, note that for the terms of $\llbracket O_{A},O_{B}\rrbracket_-  $ which do not contain $\uD$, we have		
	\[
	e_A e_B - (-1)^{|A||B|} e_Be_A = 0.
	\]	
For the terms with a single occurrence of $\uD$, we have (up to a factor $-\epsilon/4$)
\begin{align*}
&(\uD\, \ux_A e_A - e_A\ux_A   \uD  ) e_B   - (-1)^{|A||B|}  e_B (\uD\, \ux_A e_A -  e_A\ux_A  \uD)
 \\
& 
- (-1)^{|A||B|}(\uD\, \ux_B e_B - e_B \ux_B  \uD  ) e_A   + e_A (\uD\, \ux_B e_B - e_B \ux_B  \uD  )
\\
	= \ & \sum_{a\in A} \big(\uD\, x_a e_a e_Ae_B
		-   x_a e_A e_a   \uD   e_B
		- (-1)^{|A||B|}  e_B \uD\,   x_a e_ae_A 
		+(-1)^{|A||B|}e_Be_A  e_ax_a \uD  \big)
	\\
	&
		- (-1)^{|A||B|} \uD\, \ux_B e_Be_A    
		+ (-1)^{|A||B|}e_B \ux_B  \uD   e_A
		+ e_A \uD\, \ux_B e_B   
	- e_Ae_B \ux_B  \uD 
	 \\
	= \ & \sum_{a\in A} \big(\uD\, x_a e_a e_Ae_B
	-   (-1)^{|A|-1} x_a    \uD e_A e_a e_B
	- (-1)^{|A||B|+|A|-1}  	e_B e_a e_A\uD\,   x_a 
		+ e_Ae_Be_a x_a \uD    \big)
		\\
		& +  (-1)^{|A|-1} \sum_{a\in A}   x_a\sum_{l\in A\setminus\{a\}}2p_le_l e_A e_a     e_B 
		- (-1)^{|A||B|}\sum_{a\in A}   e_B \sum_{l\in A\setminus\{a\}}2p_le_l  e_a e_A   x_a\\	
		&
		-   \uD\, \ux_B e_A e_B 
		+ (-1)^{|A||B|+|A|}e_B \ux_B e_A \uD 
			+ (-1)^{|A|}\uD\,e_A \, \ux_B e_B   
		- e_Ae_B \ux_B  \uD   
		\\
	& 	+ (-1)^{|A||B|}e_B \ux_B \sum_{a\in A}2p_ae_a   e_A 
		- (-1)^{|A|} \sum_{a\in A}2p_ae_a   e_A\ux_B e_B
 \\
	= \ & \sum_{a\in A} \big(\uD\, x_a e_a e_Ae_B
	-    x_a    \uD e_ae_A  e_B
	- (-1)^{|B|+|A|-1}  	 e_a e_Ae_B\uD\,   x_a
		+(-1)^{|A|+|B|-1} e_ae_Ae_B x_a \uD     \big)
	\\
	& + (-1)^{|A|+|B|-1} \sum_{a\in A}   e_a  e_A     e_B \sum_{l\in A\setminus\{a\}}x_le_l\epsilon(\uD \, e_a + e_a \uD)
	+ \sum_{a\in A}  \epsilon(\uD \, e_a + e_a \uD)\sum_{l\in A\setminus\{a\}}e_l     x_l  e_ae_Ae_B  \\
	&	+  \sum_{a\in A} (-1)^{|B|+|A|-1}e_a   e_A e_B\ux_B\epsilon(\uD \, e_a + e_a \uD) 
	 +\sum_{a\in A}\epsilon(\uD \, e_a + e_a \uD)\ux_B e_a   e_A e_B
 \\
	= \ & \sum_{a\in A} \big(\uD\, x_a e_a e_Ae_B
	-    x_a    \uD e_ae_A  e_B	
	- (-1)^{|B|+|A|-1}  	 e_a e_Ae_B\uD\,   x_a 
	+(-1)^{|A|+|B|-1} e_ae_Ae_B x_a \uD   
	\\
	& + \epsilon(-1)^{|A|+|B|-1}    e_a  e_A     e_B \ux_{a\triangle A\triangle B}\uD \, e_a  
	- \epsilon    e_a e_a e_A     e_B \ux_{a\triangle A\triangle B}\uD 	\\
	&
		- \epsilon(-1)^{|A|+|B|-1} \uD\,   \ux_{a\triangle A\triangle B}  e_a e_A     e_B e_a
			+ \epsilon e_a\uD\,  \ux_{a\triangle A\triangle B}    e_a e_A     e_B,
\end{align*}
where we made use of \eqref{DeA}, $\{\uD , e_a\}= \epsilon 2p_a$ and $\ux_{a\triangle A\triangle B} = \ux_{A\setminus\{a\}} + \ux_B$ for $a\in A$ and $A\cap B=\emptyset$. 	
	
Next, we work out the terms of $\llbracket O_{A},O_{B}\rrbracket_-  $ containing two occurrences of $\uD$, that is,
\begin{align*}
& \big(\uD\, \ux_A e_A - e_A\ux_A   \uD \big)
\big(\uD\, \ux_B e_B - e_B\ux_B   \uD \big)  - (-1)^{|A||B|}  \big(\uD\, \ux_B e_B - e_B\ux_B   \uD\big)
\big(\uD\, \ux_A e_A -  e_A\ux_A  \uD  \big).
\end{align*}
For the terms having $\uD\uD$ in the middle, we readily find
\begin{align*}
& -e_A\ux_A   \uD\uD\, \ux_B e_B   + (-1)^{|A||B|}  e_B\ux_B   \uD\uD\, \ux_A e_A\\
= \  & -\sum_{a\in A} \big(e_A x_a e_a   \uD\uD\, \ux_B e_B  
 - (-1)^{|A||B|}  e_B\ux_B   \uD\uD\, x_a e_a  e_A\big)\\
= \  & -\sum_{a\in A} \big( x_a    \uD\uD\, \ux_B e_ae_Ae_B   
- (-1)^{|A|-1+|B|}  e_a  e_Ae_B\ux_B   \uD\uD\, x_a \big)\\
= \  & -\sum_{a\in A} \big( x_a    \uD\uD\, \ux_{a\triangle A\triangle B}  e_ae_Ae_B   
- (-1)^{|A|-1+|B|}  e_a  e_Ae_B\ux_{a\triangle A\triangle B}    \uD\uD\, x_a \big),
\end{align*}
since
\begin{align*}
  & \sum_{a\in A} \big( x_a    \uD\uD\, \ux_{a\triangle A}  e_ae_Ae_B   
- (-1)^{|A|-1+|B|}  e_a  e_Ae_B\ux_{a\triangle A}    \uD\uD\, x_a \big)\\
= \ &  \sum_{a\in A} \sum_{l\in A\setminus\{a\}} x_a    \uD\uD\, x_{l}e_l  e_a e_Ae_B  
-  \sum_{a\in A} \sum_{l\in A\setminus\{a\}} x_{l}   \uD\uD\, x_a  e_a e_le_Ae_B \\
= \ & 0.
\end{align*}
Similarly, we have
\begin{align*}
&  -\uD\, \ux_A e_A e_B\ux_B   \uD  + (-1)^{|A||B|}\uD\, \ux_B e_B e_A\ux_A   \uD\\
= \  & -\sum_{a\in A} \big(\uD\, x_a e_a  e_A e_B\ux_B   \uD
- (-1)^{|A||B|}\uD\, \ux_B e_B   e_A x_a e_a   \uD\big)\\
= \  & -\sum_{a\in A} \big(\uD\, x_a  e_a  e_Ae_B\ux_{a\triangle A\triangle B}    \uD 
- (-1)^{|A|-1+|B|}  \uD\, \ux_{a\triangle A\triangle B}  e_ae_Ae_B x_a    \uD \big).
\end{align*}
Finally, we manipulate the remaining terms as follows
\begin{align*}
& \uD\, \ux_A e_A \uD\, \ux_B e_B 
 + e_A\ux_A   \uD e_B\ux_B   \uD  
 - (-1)^{|A||B|} \uD\, \ux_B e_B\uD\, \ux_A e_A 
 - (-1)^{|A||B|}  e_B\ux_B   \uD e_A\ux_A  \uD \\
 = \ & \sum_{a\in A} \big( \uD\, x_a e_a e_A \uD\, \ux_B e_B 
 + e_Ax_a e_a   \uD e_B\ux_B   \uD  
 - (-1)^{|A||B|} (\uD\, \ux_B e_B\uD\, x_a e_ae_A 
 +  e_B\ux_B   \uD e_Ax_a e_a  \uD)\big)\\
 = \ & \sum_{a\in A} \big( (-1)^{|A|-1}\uD\, x_a \uD\, e_a e_A \ux_B e_B 
 + (-1)^{|A|-1}x_a   \uD e_A e_ae_B\ux_B   \uD  \\*
 &
 - (-1)^{|A||B|+|A|-1} \uD\, \ux_B e_Be_ae_A\uD\, x_a  
 - (-1)^{|A||B|+|A|-1}   e_B\ux_B  e_Ae_a \uD x_a   \uD\big)\\*
 & +\sum_{a\in A}\sum_{l\in A\setminus\{a\}}  
 \big( -(-1)^{|A|-1} \uD\, x_a 2p_le_l e_a e_A  \ux_B e_B 
 -(-1)^{|A|-1} x_a  2p_le_le_Ae_a   e_B\ux_B   \uD  \\*
 &
 - (-1)^{|A||B|} \uD\, \ux_B e_B 2p_le_le_ae_A  x_a 
  - (-1)^{|A||B|} e_B\ux_B    2p_le_l e_Ae_a  x_a  \uD\big)\\*
  = \ & \sum_{a\in A} \big( \uD\, x_a \uD\,\ux_{a\triangle A\triangle B}  e_a e_Ae_B  
  + x_a   \uD e_ae_A e_B\ux_{a\triangle A\triangle B}   \uD  \\
  &
  - (-1)^{|B|+|A|-1} \uD\, \ux_{a\triangle A\triangle B} e_ae_Ae_B\uD\, x_a  
  - (-1)^{|A|-1+|B|}   e_a e_Ae_B  \ux_{a\triangle A\triangle B}\uD x_a   \uD\big).
\end{align*}
In the last step, we used the following (and a similar result for the other two terms)
\begin{align*}
& \sum_{a\in A}  \uD\, x_a \uD\,\sum_{l\in A\setminus\{a\}}  x_le_l  e_a e_Ae_B  
- (-1)^{|B|+|A|-1} \sum_{a\in A} \uD\, \sum_{l\in A\setminus\{a\}}  x_le_l e_ae_Ae_B\uD\, x_a \\
= \  &\sum_{a\in A} \sum_{l\in A\setminus\{a\}} \uD\, x_a \uD\,  x_le_l  e_a e_Ae_B  
+ \sum_{a\in A}  \sum_{l\in A\setminus\{a\}}  \uD\,x_l\uD\,  e_l e_ae_Ae_Bx_a\\
& -\sum_{a\in A}  \sum_{l\in A\setminus\{a\}}  \uD\,x_l\sum_{b\in A\setminus\{a,l\}\cup B} 2p_b e_b  e_l e_ae_Ae_Bx_a\\
= \  & -\sum_{a\in A}  \sum_{l\in A\setminus\{a\}}  \uD\,x_l\sum_{b\in A\setminus\{a,l\}} 2p_b e_b  e_l e_ae_Ae_Bx_a-\sum_{a\in A}  \sum_{l\in A\setminus\{a\}}  \uD\,x_l\sum_{b\in  B} 2p_b e_b  e_l e_ae_Ae_Bx_a\\
= \  & -\sum_{\{a,l,b\}\subset A}  2\uD\,(x_l p_bx_a-x_a p_bx_l+x_b p_ax_l-x_b p_lx_a+x_a p_lx_b-x_l p_ax_b) e_b  e_l e_ae_Ae_B\\
&-\sum_{\{a,l\}\subset A}\sum_{b\in  B} 2\uD\,(x_l p_b x_a-x_a p_b x_l) e_b  e_l e_ae_Ae_B\\
= \  
&-\sum_{\{a,l\}\subset A}\sum_{b\in  B} 2\uD\,(x_l (p_b x_a-x_ap_b)-x_a(p_b x_l-x_lp_b)) e_b  e_l e_ae_Ae_B\\
= \  
&-\sum_{\{a,l\}\subset A}\sum_{b\in  B} 2\uD\,(x_l (p_a x_b-x_bp_a)-x_a(p_l x_b-x_bp_l)) e_b  e_l e_ae_Ae_B\\
= \  
&-\sum_{\{a,l\}\subset A}\sum_{b\in  B} 2\uD\,(x_lp_a-x_ap_l) x_b e_b  e_l e_ae_Ae_B-\sum_{\{a,l\}\subset A}\sum_{b\in  B}x_b(x_a p_l-x_lp_a) e_b  e_l e_ae_Ae_B\\
= \  
&\sum_{a\in A}\sum_{l\in A\setminus\{a\}} \uD\,x_a2p_l \ux_B e_l e_ae_Ae_B
-\sum_{a\in A}\sum_{l\in A\setminus\{a\}}\uD\,\ux_B2p_lx_a    e_l e_ae_Ae_B.
\end{align*}
To arrive at this result, we made use of property~\eqref{DeA},  Lemma~\ref{lemma2}, the commutativity of $x_1,\dotsc,x_N$ and $C_{ij}=C_{ji}$.

Putting everything together and comparing with 
\begin{align*}	
&   \sum_{a\in A} \llbracket O_a, O_{a, A, B}\rrbracket_-   = \sum_{a\in A}\frac{\epsilon}{4}\big((\uD\, x_{a}  -  x_{a}  \uD - e_a )
(\uD\, \ux_{a\triangle A\triangle B} e_ae_{A}e_B -e_ae_{A}e_B  \ux_{a\triangle A\triangle B}  \uD - \epsilon e_a e_{A}e_B )\\
& -(-1)^{|A|+|B|-1 }\frac{\epsilon}{4}
(\uD\, \ux_{a\triangle A\triangle B} e_ae_{A}e_B -e_ae_{A}e_B  \ux_{a\triangle A\triangle B}  \uD - \epsilon e_a e_{A}e_B )
(\uD\, x_{a}  -  x_{a}  \uD - e_a )\big),
\end{align*}	
the proof is completed.	
\end{proof}

\begin{theo}\label{theoAsubsetB}
	In the algebra $\cA\otimes \mathcal{C}$, for two lists of $\{1,\dotsc,N\}$, denoted by $A$ and $B$, such that $A\subset B$ as sets, we have
	\[
	\llbracket O_{A},O_{B}\rrbracket_-  = \epsilon \sum_{a\in A} \llbracket O_a, O_{a, A, B}\rrbracket_-   \, .
	\]	
\end{theo}
Note that in this case  ${a\triangle A\triangle B}= {\{a\}\cup (B\setminus A ) }$.
\begin{proof}
To prove this result, one is not limited to just the form~\eqref{OA} for $O_A$ and $O_B$ as we did in the proof of Theorem~\ref{theoAdisjB}. One may also use, for instance, the form~\eqref{OA2}, and employ a strategy similar to the one used in the proof of Theorem~\ref{theoAlgOii}. A proof in this style can be found in Appendix~\ref{sec:5}. 	
\end{proof}

\begin{theo}\label{theoOAOB}
	In the algebra $\cA\otimes \mathcal{C}$, for two lists of $\{1,\dotsc,N\}$, denoted by $A$ and $B$, we have
	\[
		\llbracket O_{A},O_{B}\rrbracket_+ = \epsilon O_{A, B}+ (e_{A\cap B})^2\llbracket O_{A, (A\cap B)},O_{(A\cap B),  B}\rrbracket_++ (e_{A\cap B})^2  \llbracket O_{A\cap B},O_{(A\cap B) , A, B}\rrbracket_+  .	
	\]
\end{theo}
Note that
$
{A\triangle B} = {(A\setminus B)\cup(B\setminus A)}$, while ${A\triangle (A\cap B)} = {A\setminus B}$ and ${(A\cap B)\triangle  B} = {B\setminus A}$, and finally ${(A\cap B) \triangle A\triangle B} = {A\cup B}$.

\begin{proof}	
Because of its length and as it employs a similar strategy as used already in the proof of Theorem~\ref{theoAdisjB}, we have moved the proof of this result to Appendix~\ref{sec:5}.
\end{proof}

\begin{coro}\label{coro}
	In the algebra $\cA\otimes \mathcal{C}$, for two lists of $\{1,\dotsc,N\}$, denoted by $A$ and $B$, we have	
\begin{align*}
		\llbracket O_{A},O_{B}\rrbracket_+ 
		 = \ & \epsilon O_{A, B}+ (e_{A\cap B})^2 2 O_{A, (A\cap B)}O_{(A\cap B),  B}+ (e_{A\cap B})^2  2 O_{A\cap B}O_{(A\cap B) , A, B}\\
		 & - \epsilon(e_{A\cap B})^2 \sum_{a\in A\setminus B} \llbracket O_a,O_{(A\cap B),  B}\rrbracket_-
		 - \epsilon(e_{A\cap B})^2  \sum_{a\in A\cap B} \llbracket O_a,O_{(A\cap B) , A, B}\rrbracket_- \,.
\end{align*}
\end{coro}
\begin{proof}
	Note first that $	\llbracket O_{A},O_{B}\rrbracket_- +\llbracket O_{A},O_{B}\rrbracket_+ = 2  O_AO_B $. 
Now, combine	Theorem~\ref{theoOAOB} with Theorem~\ref{theoAdisjB} and Theorem~\ref{theoAsubsetB}, using $(A\setminus B)\cap (B\setminus A)=\emptyset$ and $(A\cap B)\subset (A\cup B)$.
\end{proof}

By means of Theorem~\ref{theoAdisjB}, Theorem~\ref{theoAsubsetB} and Theorem~\ref{theoOAOB} (or thus Corollary~\ref{coro}), we can swap any two operators $O_A$ and $O_B$ where $A$ or $B$ is not a list of just one element. We briefly explain the need for three such theorems.  
Theorem~\ref{theoOAOB} yields an empty identity in two cases, when $A\cap B=\emptyset$ or when either $A$ or $B$ is contained in the other as sets. For example, say $A\cap B=\emptyset$, then we have
\begin{align*}
	\llbracket O_{A},O_{B}\rrbracket_+ =\ & \epsilon O_{A, B}+ (e_{A\cap B})^2\llbracket O_{A, (A\cap B)},O_{(A\cap B),  B}\rrbracket_+ + (e_{A\cap B})^2  \llbracket O_{A\cap B},O_{(A\cap B) , A, B}\rrbracket_+ \\	
	=\ & \epsilon O_{A, B}+ \llbracket O_{A},O_{ B}\rrbracket_+ +  O_{\emptyset}O_{A, B}+ O_{A, B}O_{\emptyset}\\	
	=\ & \llbracket O_{A},O_{ B}\rrbracket_+ \,,
\end{align*}
as $e_{\emptyset}= 1$ and $O_{\emptyset}= -\epsilon/2$. For these cases, we can resort to Theorem~\ref{theoAdisjB} or Theorem~\ref{theoAsubsetB}.
However, if $A$ is a list of a single element $a$, Theorem~\ref{theoAdisjB} and Theorem~\ref{theoAsubsetB} are empty identities:
\[	
\llbracket O_{A},O_{B}\rrbracket_-  = 
\epsilon \sum_{a\in A} \llbracket O_a, O_{a, A, B}\rrbracket_-   = 
\epsilon \llbracket O_a, O_{a, a, B}\rrbracket_-    = 
 \llbracket O_a, O_{ B}\rrbracket_-  ,
\]
but so is Theorem~\ref{theoOAOB} as either $a\in B$ or $a\cap B =\emptyset$. Now, for the case $a\notin B$ Theorem~\ref{OAa} yields
\begin{equation*}
O_a O_B = (-1)^{|B|}O_BO_a - \epsilon
\sum_{b\in B} \llbracket O_{b},O_{b,a, B}\rrbracket_-  \,  , 
\end{equation*}
while by definition we also have
\[
O_aO_B = \pm(-1)^{|B|-|a\cap B|}O_BO_a +\llbracket O_{a},O_{B}\rrbracket_{\mp}.
\]

We see that all expressions involving supercommutators can be reduced to sums of supercommutators containing (at least) one symmetry with just one index. 
In the following section, we review again examples of specific realizations of the elements $x_1,\dotsc,x_N$ and $p_1,\dotsc,p_N$ of the algebra $\cA$. For these examples, the one-index symmetries in particular take on an explicit form whose interaction with other symmetries can be computed explicitly. This form depends on the makeup of the symmetries $C_{ij}=[p_i,x_j]$ as $O_i$ is given by \eqref{Oi}.

\section{Examples}

We recall the two examples from the previous section.

\noindent\textbf{Example 4.1.} Let again $\Delta$ be the classical Laplace operator, then
\[
\uD = \sum_{i=1}^N e_i p_i= \sum_{i=1}^N e_i \frac{\partial}{\partial x_i}
\]
is the classical Dirac operator. The $\mathfrak{osp}(1|2)$ structure of Theorem~\ref{theoosp} here was obtained already in~\cite{HS}. 
The commutator $C_{li}$ in the definition~\eqref{Oi} of $O_i$ reduces to
$
[p_i,x_j] =   \delta_{ij},
$ so
\[
O_{i}    =  \epsilon \frac12\left(\sum_{l=1}^N  e_l \delta_{li}-   e_i \right)= \epsilon \frac12\left( e_i-   e_i \right) = 0.
\]
Moreover, \eqref{Oij} then becomes
\[
O_{ij} 	 =L_{ij}+\epsilon\frac12 e_ie_j  ,
\]
in accordance with results obtained in~\cite{DSS}, while for a general subset $A\subset \{1,\dotsc,N\}$ one has
\[
O_A = \bigg(\epsilon\frac{|A|-1}{2}- \sum_{\{i, j\} \subset A} L_{ij} e_ie_j\bigg)e_A.
\]
Since 	$
[p_i,x_j] =   0
$ for $i\neq j$, given a subset $A\subset \{1,\dotsc,N\}$ the operators $\ux_A$ and $\uD_A$ as defined by~\eqref{DAxA} in fact also generate a copy of the Lie superalgebra $\mathfrak{osp}(1|2)$ whose Scasimir element we will denote by 
\[
\mathcal{S}_A   = \frac12\left( [\uD_A,\ux_A]-\epsilon\right)  .
\] 
From \eqref{OA}, we see that in this case $
O_A$ equals $ \mathcal{S}_A e_A$. 

As in this case the one-index symmetries are identically zero, the algebraic relations simplify accordingly. The symmetries with two indices generate a realization of the Lie algebra $\mathfrak{so}(N)$, as seen from the relations of Theorem~\ref{theoAlgOii}. For general symmetries, 
Theorem~\ref{theoOAOB} now yields
\[
\llbracket O_{A},O_{B}\rrbracket_+ = \epsilon O_{A, B}+ (e_{A\cap B})^2 2 O_{A, (A\cap B)}O_{(A\cap B),  B}+ (e_{A\cap B})^2  2 O_{A\cap B}O_{(A\cap B) , A, B},
\]
as by Theorem~\ref{theoAdisjB} and Theorem~\ref{theoAsubsetB} 
$	
\llbracket O_{A},O_{B}\rrbracket_-  = 0
$
for $A\cap B=\emptyset$, or $A\subset B$, or $B\subset A$. This corresponds to (a special case of) the higher rank Bannai--Ito algebra of \cite{DeBie&Genest&Vinet-2016-2}, which strictly speaking is not included in the results obtained there.

To conclude, we mention another group of symmetries of $\uD$ and $\ux$ which are not inside the algebra $\cA \otimes \mathcal{C}$ in this case, but which will also be useful for the next example. For a normed vector $\alpha$, its embedding in the Clifford algebra
\[
\underline{\alpha} = \sum_{l=1}^N e_l \alpha_l 
\]
is an element of the so-called $\mathrm{Pin}$ group, which forms a double cover of the orthogonal group $\mathrm{O}(N)$. These elements have the property that 
\[
\underline{\alpha}\,\underline{v}= -\underline{w}\,\underline{\alpha},
\]
where $w=\sigma_{\alpha}(v)={v} -2\langle {v},{\alpha}\rangle{\alpha}/\|{\alpha}\|^2$. Hence, if we define the operators $S_{\alpha}$ as follows 
\begin{equation}
\label{Sa}
S_{\alpha} = \sum_{l=1}^N e_l \alpha_l \sigma_{\alpha} = \underline{\alpha} \,  \sigma_{\alpha} ,
\end{equation}
then it is immediately clear that they satisfy the following properties
\[
S_{\alpha}  \underline{v} =  -\underline{w} S_{\alpha},\qquad S_{\alpha} f(x)  =   f(\sigma_{\alpha}(x)) S_{\alpha},\qquad (S_{\alpha})^2 =\epsilon. 
\]
where again $w=\sigma_{\alpha}(v)$ and $f$ is a function or operator which does not interact with the Clifford generators. 
By direct computation, one can show that the operator $S_{\alpha}$ anticommutes with the Dirac operator. Moreover, the interaction of $S_{\alpha}$ and a symmetry operator $O_A$ is simply given by the action of the reflection associated with $\alpha$ on the coordinate vectors corresponding to the elements of $A$.

\noindent\textbf{Example 4.2.} We consider again the case where $p_1,\dotsc,p_N$ are given by the Dunkl operators~\eqref{Dunkl} associated with a given root system $R$ and with multiplicity function $k$.
Here, the $\mathfrak{osp}(1|2)$ structure of Theorem~\ref{theoosp} was obtained already in~\cite{2012_DeBie&Orsted&Somberg&Soucek_TransAmerMathSoc_364_3875,
	2009_Orsted&Somberg&Soucek_AdvApplCliffAlg_19_403}.

The commutator in the definition~\eqref{Oi} of $O_i$ is then given by
\[
C_{ij} = [\cD_i,x_j] = \delta_{ij} +  \sum_{\alpha\in R_+} 2k(\alpha) \alpha_i \alpha_j \sigma_{\alpha}.
\]
The symmetries of the Dunkl Dirac operator $\sum_{i=1}^N\cD_ie_i$ with one index thus become
\[
O_i = \epsilon \frac12\left(\sum_{l=1}^N  e_l \delta_{li} +  \sum_{l=1}^N  e_l\sum_{\alpha\in R_+} 2k(\alpha) \alpha_l \alpha_i \sigma_{\alpha} -   e_i\right) = \epsilon   \sum_{\alpha\in R_+} k(\alpha)\alpha_i  \sum_{l=1}^N  e_l\alpha_l \sigma_{\alpha} . 
\]
On the right-hand side we see the operators~\eqref{Sa} appear for the roots $\alpha\in R_+$. 
By direct computation, one can show that for a root $\alpha$, the operator $S_{\alpha}$ anticommutes with the Dirac--Dunkl operator. 
The one-index symmetry $O_i$ thus consists of linear combinations of the operators~\eqref{Sa} determined by the root system and by the multiplicity function $k$
\[
O_i =  \epsilon \sum_{\alpha\in R_+} k(\alpha)\alpha_iS_{\alpha}.
\]
Higher-index symmetries $O_A$ contain the Dunkl angular momentum operators, appended with the anticommuting symmetries $S_{\alpha}$ for $\alpha\in R_+$. For instance, if $A=\{i,j\}$ we have
\begin{equation*}
O_{ij}  = x_i \cD_j - x_j \cD_i + \epsilon\frac12 e_i e_j + \sum_{\alpha\in R_+} k(\alpha) (  e_i \alpha_j   - e_j\alpha_i  )S_{\alpha} .
\end{equation*}

The algebraic relations of Theorem~\ref{theoAdisjB} and Theorem~\ref{theoAsubsetB} can now be worked out explicitly as the action of the one-index symmetries is given by the reflections associated with the roots of the root system.

\noindent\textbf{Example 4.2.1.} For the root system with Weyl group $W =(\mathbb{Z}_2)^N$, our results are in accordance with what was already obtained in \cite{DeBie&Genest&Vinet-2016,DeBie&Genest&Vinet-2016-2}. Here, the Dunkl operators~\eqref{Dunkl} are given by
\[
\cD_i = \frac{\partial}{\partial x_i} + \frac{\mu_i}{x_i}(1-r_i) \qquad i\in\{1,\dotsc,N\}
\]
where $r_i$ is the reflection operator in the $x_i=0$ hyperplane and $\mu_i$ is the value of the multiplicity function on the conjugacy class of $r_i$. 
The one-index symmetry~\eqref{Oi} in this case reduces to
\[
O_i =  \epsilon \sum_{\alpha\in R_+} k(\alpha)\alpha_iS_{\alpha}
=  \epsilon \mu_i r_i e_i.
\]
Here, we also have $
[p_i,x_j] =   0
$ for $i\neq j$, so for a given subset $A\subset \{1,\dotsc,N\}$ the operators $\ux_A$ and $\uD_A$ as defined by~\eqref{DAxA} generate a copy of the Lie superalgebra $\mathfrak{osp}(1|2)$ with the Scasimir element 
\[
\mathcal{S}_A   = \frac12\left( [\uD_A,\ux_A]-\epsilon\right)  .
\] 
From \eqref{OA}, we see that in this case $
O_A$ equals $ \mathcal{S}_A e_A$. 
The relation with the symmetries denoted by $\Gamma_A$ in \cite{DeBie&Genest&Vinet-2016-2} is 
\[
\Gamma_A = \mathcal{S}_A \prod_{i\in A} r_i = \mathcal{S}_A \prod_{i\in A} \frac{1}{\mu_i}O_ie_i= \mathcal{S}_A e_A\prod_{i\in A} \frac{1}{\mu_i}O_i= O_A \prod_{i\in A} \frac{1}{\mu_i}O_i,
\]
where the product over $i\in A$ is taken according to the order of $A$. The operator $\Gamma_A$ commutes with the Dunkl Dirac operator by Corollary~\ref{cor}. The algebraic structure generated by the operators $\Gamma_A$ and corresponding to the relations of Theorem~\ref{theoOAOB} is the higher rank Bannai--Ito algebra of \cite{DeBie&Genest&Vinet-2016-2}. For the case $N=3$, see \cite{DeBie&Genest&Vinet-2016,2015_Genest&Vinet&Zhedanov_CommMathPhys_336_243}, this is the regular Bannai--Ito algebra~\cite{2012_Tsujimoto&Vinet&Zhedanov_AdvMath_229_2123}.

\noindent\textbf{Example 4.2.2.} For the root system of type $A_{N-1}$, with positive subsystem given, for instance, by
\[
R_+ = \Big\{  \frac{1}{\sqrt{2}}(\xi_i-\xi_j)   \Bigm\vert 1\leq i<j\leq N  \Big\} 
\]
where  $\{\xi_1,\dotsc,\xi_N\}$ is an orthonormal basis of $\mathbb{R}^N$, the associated Weyl group is the symmetric group $S_N$ of permutations on $N$ elements. All permutations in $S_N$ are conjugate so the multiplicity function $k(\alpha)$ has the same value for all roots, which we will denote by $\kappa$. 
The Dunkl operators~\eqref{Dunkl} are then given by
\[
\mathcal{D}_i = \frac{\partial}{\partial x_i} + \kappa  \sum_{j\neq i} \frac{1-g_{ij}}{x_i-x_j} \qquad i\in\{1,\dotsc,N\}
\]
where $g_{ij}$ denotes the reflection corresponding to the root ${1/\sqrt{2}}(\xi_i-\xi_j)$. 
The related operator of the form~\eqref{Sa} will be denoted as
\[
G_{ij} = \frac{1}{\sqrt{2}} (e_i - e_j) g_{ij} = -G_{ji}.
\]

In this case, the commutator $[p_i,x_j]$ does not reduce to zero for $i\neq j$. It is given by
\[
	[\mathcal{D}_i,x_j] =\begin{cases} 1+\kappa\displaystyle\sum_{l\neq i}g_{il} & \qquad \text{if }i=j \\ -\kappa g_{ij} & \qquad \text{if } i\neq j \end{cases}
\]
and the one-index symmetry~\eqref{Oi} becomes	
\[
O_i =  \epsilon \sum_{\alpha\in R_+} k(\alpha)\alpha_iS_{\alpha}
=  \epsilon \kappa\sum_{1\leq l<i} \frac{-1}{\sqrt{2}}G_{li}
+\epsilon \kappa\sum_{i<l\leq N} \frac{1}{\sqrt{2}}G_{il}=  \frac{\epsilon \kappa}{\sqrt{2}}\sum_{l=1}^N G_{il},
\]
where $G_{ii}=0$, while for the symmetry~\eqref{Oij} we have
\begin{align*}
O_{ij}  = x_i \cD_j - x_j \cD_i + \epsilon\frac12 e_i e_j + \frac{\epsilon\kappa}{\sqrt{2}} \sum_{l=1}^N (  e_i G_{jl}   - e_j G_{il}  ).
\end{align*}

The relations for two-index symmetries of Theorem~\ref{theoAlgOii}
are
\begin{align*}
[O_{ij},O_{kl}]  &=     \frac{\kappa}{\sqrt{2}}\big( (O_{lij}-O_{lik})(G_{jk}) +(O_{kji}-O_{kjl})(G_{il})+(O_{ikl}-O_{ikj})(G_{lj})+(O_{jlk}-O_{jli})(G_{ki}) \big)
\end{align*}
for four distinct indices, and when $l=i$ we have
\begin{align*}
& [O_{ij},O_{ki}] =   O_{jk} +  \frac{\kappa}{\sqrt{2}} 2O_{ijk}(G_{ij}-G_{ki})+  \frac{\kappa}{\sqrt{2}} \sum_{a\neq i,j,k} (O_{ijk} - O_{ajk})G_{ia} +\frac{\kappa^2}{2} \sum_{a=1}^N\sum_{b=1}^N [G_{ja},G_{kb}].
\end{align*}

\section{Summary and outlook}

The replacement of ordinary derivatives by Dunkl derivatives ${\cal D}_i$ in the expressions of the Laplace and the Dirac operator in $N$ dimensions gives rise to the Laplace--Dunkl $\Delta$ and Dirac--Dunkl operator $\underline{D}$. 
This paper was devoted to the study of the symmetry algebras of these two operators, i.e., to the algebraic relations satisfied by the operators commuting or anticommuting with $\Delta$ or $\underline{D}$. 

In the case of Dunkl derivatives, the underlying object is the reflection group $G$ acting in $N$-dimensional space, characterized by a reduced root system. 
The Dunkl derivatives themselves then consist of an ordinary derivative plus a number of difference operators depending on this reflection group.
So it can be expected that the reflection group $G$ plays an essential role in the structure of the symmetry algebra.

One of the leading examples was for $N=3$ and $G=\mathbb{Z}_2^3$. 
Even for this quite simple reflection group, the study of the symmetry algebras was already non-trivial, and led to the celebrated Bannai--Ito algebra~\cite{2015_Genest&Vinet&Zhedanov_CommMathPhys_336_243,DeBie&Genest&Vinet-2016}.
Following this, the second case where the symmetry algebra could be determined was for general $N$ and $G=\mathbb{Z}_2^N$~\cite{DeBie&Genest&Vinet-2016-2}, leading to a ``higher rank Bannai--Ito algebra,'' of which the structure is already highly non-trivial.

The question that naturally arises is whether the symmetry algebras of the Laplace--Dunkl and Dirac--Dunkl operators for other reflection groups $G$ can still be determined, and what their structure is.
We considered it as a challenge to study this problem.
Originally, we were hoping to solve the problem for the case $G=S_N$, which would already be a significant breakthrough.
Herein, $S_N$ is the symmetric group acting on the coordinates $x_i$ by permuting the indices; as a reflection group it is associated with the root system of type $A_{N-1}$.

Our initial attempts and computations for the case $G=S_N$ were not promising, and the situation looked extremely complicated, particularly because the explicit actions of the Dunkl derivatives ${\cal D}_i$ are already very complex. 
Fortunately, at that moment we followed the advice ``if you cannot solve the problem, generalize it.'' 
So we went back to the general case, with arbitrary reflection group $G$, and no longer focused on the explicit actions of the Dunkl derivatives ${\cal D}_i$, but on the algebraic relations among the coordinate operators $x_i$ and the ${\cal D}_i$.
Then we realized that we could still jump one level higher in the generalization, and just work with coordinate operators $x_i$ and ``momentum operators'' $p_i$ in the framework of a Wigner quantum system, by identifying the $p_i$ with ${\cal D}_i$.
As a consequence, we could forget about the actual meaning of the Dunkl derivatives, and just work and perform our computations in the associative algebra ${\cal A}$ (Definition~\ref{defA}).
This general or ``more abstract'' setting enabled us to determine the elements (anti)commuting with $\Delta$ or $\underline{D}$, and to construct the algebraic relations satisfied by these elements.
The resulting symmetry algebra, obtained in the paper, is still quite complicated.
But we managed to determine it (for general $G$), going far beyond our initial goal.
For the general Laplace--Dunkl operator, the symmetries and the symmetry algebra are described in Theorem~\ref{theoSyms} and Theorem~\ref{theoLAlg}.
For the general Dirac--Dunkl operator, the symmetries are determined in Theorem~\ref{DOi} and Theorem~\ref{TOA}. 
The relations for these symmetries (i.e., the symmetry algebra) are established in Section~\ref{ssec:3.4}, and follow from Theorems~\ref{theoAdisjB}, \ref{theoAsubsetB} and \ref{theoOAOB}.

The results of this paper open the way to several new investigations.
In particular, one could now go back to the interesting case $G=S_N$, starting with $N=3$~\cite{DBOVdJ}, and investigate how the symmetry algebra specializes. 
One of the purposes is to study representations of the symmetry algebra in that case, since this leads to null solutions of the Dirac operator.
As in the case of $G=\mathbb{Z}_2^N$\cite{2015_Genest&Vinet&Zhedanov_CommMathPhys_336_243,DeBie&Genest&Vinet-2016}, one can expect that interesting families of orthogonal polynomials should arise. Furthermore, note that for the case of $G=\mathbb{Z}_2^3$ a superintegrable model on the two-sphere was obtained~\cite{DeBie&Genest&Tsujimoto}. It is definitely worthwhile to investigate possible superintegrable systems for other groups $G$. 

In a different direction, one can examine whether the context of Wigner quantum systems, as used here for rational Dunkl operators, is still of use for other types of operators. 
Possible examples are trigonometric Dunkl operators, or the Dunkl operators appearing in discrete function theory.

Finally, it is known that solutions of Wigner quantum systems with a Hamiltonian of type~\eqref{Ham} can be described in terms of unitary representations of the Lie superalgebra $\mathfrak{osp}(1|2N)$~\cite{SV2005,VdJWigner}.
The action and restriction of the coordinate operators $x_i$ and the Dunkl operators ${\cal D}_i$ in these representations should be studied further.

\appendix

\section{Appendix}
\label{sec:5}

This appendix contains the proofs of Theorem~\ref{theoAsubsetB} and Theorem~\ref{theoOAOB}, which were omitted from the main text due to their length. 

\begin{proof}[Proof of Theorem~\ref{theoAsubsetB}]
We systematically go over every term appearing in $\llbracket O_{A},O_{B}\rrbracket_-$, using now the form~\eqref{OA2} for $O_A$ and $O_B$, that is,
\[
\left\llbracket \Big(\epsilon\frac{|A|-1}{2}+  \epsilon \sum_{a\in A} O_a e_a - \sum_{\{i, j\} \subset A} L_{ij} e_ie_j\Big) e_A,\Big(\epsilon\frac{|B|-1}{2}+  \epsilon \sum_{b\in B} O_b e_b - \sum_{\{k, l\} \subset B} L_{kl} e_ke_l\Big) e_B\right\rrbracket_-.
\]
Starting with the terms which contain no $O_i$ or $L_{ij}$,
we have using property~\eqref{eAeB} and $A\subset B$ 
\[
\frac{|A|-1}{2}\frac{|B|-1}{2}  (e_A  e_B - (-1)^{|A||B|-|A\cap B|}e_Be_A) 
=  \frac{|A|-1}{2}\frac{|B|-1}{2}  (e_A  e_B - e_Ae_B) 
=0.
\]

Next, using property~\eqref{eAeB} and Lemma~\ref{lemmA} we have 	
\begin{align*}
\left\llbracket  e_A , \sum_{b\in B} O_b e_b e_B\right\rrbracket_-  = \
& \sum_{b\in B} e_Ae_Be_bO_b -  (-1)^{|A||B|-|A\cap B|}\sum_{b\in B} O_b e_b e_B e_A  \\
= \
& \sum_{b\in B} e_Ae_Be_bO_b - \sum_{b\in B} O_b e_b  e_Ae_B  \\
= \
& \sum_{b\in A} e_Ae_Be_bO_b +\sum_{b\in B\setminus A} e_Ae_Be_bO_b  - \sum_{b\in B\setminus A}    e_Ae_Be_b O_b   - \sum_{b\in A} O_b e_b  e_Ae_B  \\
= \
& \sum_{b\in A} e_Ae_Be_bO_b - \sum_{b\in A} O_b e_b  e_Ae_B ,   
\end{align*}
while 
\begin{align*}
\left\llbracket   \sum_{a\in A} O_a e_a e_A , e_B\right\rrbracket_-   =\ &    \sum_{a\in A}O_a e_ae_Ae_B -  (-1)^{|A||B|-|A\cap B|}\sum_{a\in A}e_Be_Ae_a O_a  \\ =\ &   \sum_{a\in A}O_a e_ae_Ae_B - \sum_{a\in A}e_Ae_Be_a O_a  .   
\end{align*}
Hence, we have
\begin{align*}
& \frac{|A|-1}{2}\left\llbracket  e_A , \sum_{b\in B} O_b e_b e_B\right\rrbracket_-  
+  \frac{|B|-1}{2}\left\llbracket   \sum_{a\in A} O_a e_a e_A , e_B\right\rrbracket_-   \\
= \  &   \frac{|B|-1-(|A|-1)}{2} \left( \sum_{a\in A} O_a e_ae_Ae_B- \sum_{a\in A}e_Ae_Be_a O_a \right)  \\
= \  & \frac{|B|-|A|+1-1}{2} \sum_{a\in A}  \left( O_a  e_ae_Ae_B-(-1)^{|B|-|A|} e_ae_Ae_B O_a \right)\\
= \  &  \sum_{a\in A}  \Big\llbracket O_a ,\frac{|B|-|A|+1-1}{2} e_ae_Ae_B\Big\rrbracket_-.
\end{align*}

For the next part, using the notation $ A'= A\setminus \{a\}$, we find
\begin{align*}
\left\llbracket  \sum_{a\in A} O_a e_a e_A , \sum_{b\in B} O_b e_b e_B\right\rrbracket_-  
=  \ & \sum_{a\in A}\sum_{b\in B} (O_a e_a e_A    e_Be_b O_b -  (-1)^{|A||B|-|A\cap B|} O_b e_b e_B  e_A e_a O_a )  \\
= \ &  \sum_{a\in A}\sum_{b\in B} (O_a e_a e_A    e_Be_b O_b -  O_b e_b e_Ae_B   e_a O_a ) \\
= \ &  \sum_{a\in A} \bigg\llbracket  O_a , \sum_{b\in B\setminus A'}O_b e_b e_a e_A    e_B\bigg\rrbracket_-  \\
& + \sum_{a\in A}\sum_{b\in A'} O_a e_a e_A    e_Be_b O_b - \sum_{a\in A}\sum_{b\in A'}O_b e_b e_Ae_B   e_a O_a  .
\end{align*}
One easily sees that the summations in the last line cancel out.

Now, for the parts containing ``$L_{ij}$-terms,'' we have
\[
\left\llbracket \sum_{\{i, j\} \subset A} L_{ij} e_ie_j e_A , e_B \right \rrbracket_- = \sum_{\{i, j\} \subset A} L_{ij} \left(e_ie_j e_A  e_B -(-1)^{|A||B|-|A\cap B|} e_Be_ie_j e_A \right) = 0, 
\]
as, using property~\eqref{eAeB} and $A\subset B$, 
\begin{align*}
(-1)^{|A||B|-|A\cap B|} e_Be_ie_j e_A   =\ &  (-1)^{|A||B|-|A\cap B|+|A\setminus\{i,j\}||B|-|(A\setminus\{i,j\})\cap B|} e_ie_j e_A   e_B \\
=\ &  e_ie_j e_A   e_B.
\end{align*}

Moreover, using Lemma~\ref{lemmA} and property~\eqref{eAeB} while denoting $ A'= A\setminus \{a\}$, we have
\begin{align*}
&  \bigg\llbracket\epsilon\sum_{a\in A} O_ae_a  e_A ,  -\sum_{\{k, l\} \subset B} L_{kl} e_ke_l e_B
\bigg\rrbracket_- \\
= \ &  -\epsilon \sum_{a\in A} O_a  \sum_{\{k, l\} \subset B\setminus A'} L_{kl} e_ke_le_a  e_A  e_B
+\epsilon\sum_{a\in A}  \sum_{\{k, l\} \subset B\setminus A'} L_{kl} e_ke_l  e_Ae_Be_a   O_a\\
&  - \epsilon\sum_{a\in A} O_a  \sum_{\{k, l\} \subset A'} L_{kl} e_ke_le_a  e_A  e_B
+  \epsilon\sum_{a\in A} \sum_{\{k, l\} \subset A'} L_{kl} e_ke_l   e_Ae_B e_a O_a\\
& + \epsilon \sum_{a\in A} O_a  \sum_{k  \in A'}\sum_{l  \in B\setminus A'} L_{kl} e_ke_le_a  e_A  e_B
+ \epsilon\sum_{a\in A} \sum_{k  \in A'}\sum_{l  \in B\setminus A'} L_{kl} e_ke_l   e_Ae_Be_a  O_a\\
= \ &  - \epsilon\sum_{a\in A} \bigg\llbracket  O_a , \sum_{\{k, l\} \subset B\setminus A'} L_{kl} e_ke_le_a  e_A  e_B\bigg\rrbracket_- \\
&  - \epsilon\sum_{a\in A} O_a  \sum_{\{k, l\} \subset A'} L_{kl} e_ke_le_a  e_A  e_B
+ (-1)^{|B|-|A|} \epsilon\sum_{a\in A} \sum_{\{k, l\} \subset A'} L_{kl} e_ke_le_a   e_Ae_B  O_a\addtocounter{equation}{1}\tag{\theequation}\label{(b)}\\
& + \epsilon \sum_{a\in A} O_a  \sum_{k  \in A'}\sum_{l  \in B\setminus A'} L_{kl} e_ke_le_a  e_A  e_B \addtocounter{equation}{1}\tag{\theequation}\label{(c)}
+ (-1)^{|B|-|A|}\epsilon\sum_{a\in A} \sum_{k  \in A'}\sum_{l  \in B\setminus A'} L_{kl} e_ke_le_a   e_Ae_B  O_a
.
\end{align*}
We see that the final part to make $\epsilon \sum_{a\in A} \llbracket O_a, O_{a\triangle A\triangle B}\rrbracket_- $ appears here.  To complete the proof, we show that the last two lines, \eqref{(b)} and \eqref{(c)}, cancel out with the remaining terms of $\llbracket O_A,O_B\rrbracket_-$. 

Hereto, we use Lemma~\ref{lemmA}, to find (denoting again $ A'= A\setminus \{a\}$)
\begin{align*}
& \bigg\llbracket -\sum_{\{i, j\} \subset A} L_{ij} e_ie_j e_A ,
\epsilon\sum_{b\in B} O_b e_b e_B \bigg\rrbracket_- \\
= \  &   -\epsilon\sum_{\{i, j\} \subset A} L_{ij} e_ie_j e_A
e_B  \sum_{a\in A\setminus\{i,j\}} e_a O_a 
+(-1)^{|A||B|-|A\cap B|} \epsilon\sum_{\{i, j\} \subset A}\sum_{a\in A\setminus\{i,j\}} O_a L_{ij}e_a e_ie_j e_B
e_A \\*
&  - \epsilon\sum_{\{i, j\} \subset A} L_{ij} e_ie_j e_A
e_B  \sum_{b\in B\setminus(A\setminus\{i,j\})} e_b O_b 
+(-1)^{|A||B|-|A\cap B|} \epsilon\sum_{\{i, j\} \subset A} \sum_{b\in B\setminus(A\setminus\{i,j\})} O_b L_{ij}e_be_ie_j e_B
e_A \\
= \  &  - (-1)^{|B|-|A|} \epsilon\sum_{a\in A}\sum_{\{i, j\} \subset A'} L_{ij} e_ie_j  e_a e_A
e_B O_a 
+\epsilon\sum_{a\in A}\sum_{\{i, j\} \subset A'} O_a L_{ij} e_ie_j e_ae_A e_B
\\* & -\epsilon \sum_{\{i, j\} \subset A} L_{ij}  \sum_{b\in B\setminus(A\setminus\{i,j\})} O_b e_b e_ie_j e_A
e_B 
+\epsilon\sum_{\{i, j\} \subset A}\sum_{b\in B\setminus(A\setminus\{i,j\})} O_b L_{ij}e_b
e_ie_j e_A e_B \\
= \ & \epsilon\sum_{a\in A}O_a \sum_{\{i, j\} \subset A'}  L_{ij} e_ie_j e_ae_A e_B - (-1)^{|B|-|A|} \epsilon\sum_{a\in A}\sum_{\{i, j\} \subset A'} L_{ij} e_ie_j  e_a e_A
e_B O_a 
\\*
& -\epsilon\sum_{\{i, j\} \subset A}\sum_{b\in B\setminus(A\setminus\{i,j\})}[ L_{ij},O_b]  e_b e_ie_j e_A
e_B .\addtocounter{equation}{1}\tag{\theequation}\label{(d)}
\end{align*}
This already causes \eqref{(b)} to vanish, so \eqref{(c)} and \eqref{(d)} remain.

Next, we look at
\begin{equation}
\label{LijLklAinB}
\sum_{\{i, j\}\subset A} \sum_{\{k, l\}\subset B}L_{ij}L_{kl} e_ie_j e_A e_ke_l e_B - (-1)^{|A||B|-|A\cap B|} L_{kl}L_{ij} e_ke_l  e_Be_ie_j e_A .
\end{equation}
According to the sign of
\begin{align*}
(-1)^{|A||B|-|A\cap B|}e_ke_l e_Be_ie_j e_A   =\ &  (-1)^{|A||B|-|A\cap B|+|A\setminus\{i,j\}||B\setminus\{k,l\}|-|(A\setminus\{i,j\})\cap (B\setminus\{k,l\})|} e_ie_j e_A e_ke_l  e_B \\*
=\ &  (-1)^{|A\cap B|-|(A\setminus\{i,j\})\cap (B\setminus\{k,l\})|} e_ie_j e_A e_ke_l  e_B,
\end{align*}
the summation~\eqref{LijLklAinB} reduces to a combination of commutators and anticommutators involving $L_{ij}$ and $L_{kl}$. 
We first treat the anticommutators
\begin{align*}
&  - \sum_{k \in  A}\sum_{\{i, j\} \subset  A\setminus\{k\} }\sum_{l\in  B\setminus A  }\{L_{ij},L_{kl}\} e_ie_je_ke_l e_A  e_B \\
& - \sum_{i \in A}\sum_{\{j, k\} \subset  A\setminus\{i\} }\epsilon(\{L_{ik},L_{ji}\} e_ke_j e_A  e_B+\{L_{ij},L_{ki}\} e_je_k e_A  e_B).
\end{align*}
The last line reduces to
\begin{align*}
&   \sum_{i \in A}\sum_{\{j, k\} \subset  A\setminus\{i\} }\epsilon(\{L_{ik},L_{ji}\} -\{L_{ij},L_{ki}\}) e_ke_j e_A e_B =0, 
\end{align*}
while the first line vanishes by Theorem~\ref{theoLAlg} as it can be rewritten as
\begin{align*}
&   \sum_{\{k,i, j\} \subset  A }\sum_{l\in  B\setminus A  }(\{L_{ij},L_{kl}\} e_ie_je_ke_l e_A  e_B+\{L_{kj},L_{il}\} e_ke_je_ie_l e_A  e_B+\{L_{ik},L_{jl}\} e_ie_ke_je_l e_A  e_B) \\
= \ & \sum_{\{k,i, j\} \subset  A }\sum_{l\in  B\setminus A  }(\{L_{ij},L_{kl}\} +\{L_{jk},L_{il}\} +\{L_{ki},L_{jl}\} )e_ie_je_ke_l e_A  e_B.
\end{align*}

Next, we treat the remaining terms of the summation~\eqref{LijLklAinB} which reduce to four different summations of commutators
\begin{equation}
\begin{array}{cl}
& \displaystyle\sum_{ \{i, j\} \subset A }  \sum_{\{k, l\} \subset B\setminus A}[L_{ij},L_{kl}] e_ie_je_ke_l e_A  e_B 
+\sum_{ \{i, j\} \subset A} \sum_{ \{k, l\} \subset (B\cap A)\setminus\{i,j\}}[L_{ij},L_{kl}] e_ie_je_ke_l e_A  e_B\\
&  \displaystyle
-  \sum_{ i \in A}\sum_{ j \in A\setminus\{i\}}^{\phantom{l}}	\sum_{l\in  B\setminus A  }[L_{ij},L_{jl}] e_ie_je_je_l e_A  e_B +\sum_{ \{i, j\} \subset A}[L_{ij},L_{ij}] e_ie_je_ie_j e_A  e_B.
\end{array} \label{(e)}
\end{equation}
Note that while the last summation obviously vanishes, the second one does also as
\begin{align*}
& \sum_{ \{i, j\} \subset A\cap B} \sum_{ \{k, l\} \subset (B\cap A)\setminus\{i,j\}}[L_{ij},L_{kl}] e_ie_je_ke_l e_A  e_B \\ = \ &	\sum_{ \{i,j,k,l\} \subset A\cap B} ([L_{ij},L_{kl}]+[L_{ik},L_{lj}]+[L_{il},L_{jk}]+[L_{jk},L_{il}]+[L_{jl},L_{ki}]+[L_{kl},L_{ij}]) e_ie_je_ke_l e_A  e_B .
\end{align*}
By Theorem~\ref{theoLAlg} and using $\{e_i,O_j\} = [p_i,x_j] - \delta_{ij}$, the summand of the first summation of \eqref{(e)} can be written as
\begin{align*}
& L_{jk}\{e_i,O_l\}e_je_ke_ie_l e_A  e_B+	L_{il}\{e_j,O_k\}e_ie_le_je_k e_A  e_B
\\ & L_{lj}\{e_k,O_i\}e_le_je_ie_k e_A  e_B+L_{ki}\{e_l,O_j\}e_ke_ie_je_l e_A  e_B.
\end{align*}
Using $\{e_j,O_k\}=\{e_k,O_j\}$ and $\{e_i,O_l\}=\{e_l,O_i\}$, each term is altered to one containing $O_b$ where $b\in B\setminus A$. In this way, we find
\begin{align*}
&	 \epsilon(L_{jk}O_le_je_ke_l e_A  e_B+L_{lj}O_ke_le_je_k e_A  e_B+L_{il}O_ke_ie_le_k e_A  e_B+L_{ki}O_le_ke_ie_l e_A  e_B )\\
& +	L_{jk}e_iO_le_je_ke_ie_l e_A  e_B+L_{lj}e_iO_ke_le_je_ie_k e_A  e_B +L_{il}e_jO_ke_ie_le_je_k e_A  e_B+L_{ki}e_jO_le_ke_ie_je_l e_A  e_B.
\end{align*}
The summation, as in~\eqref{(e)}, of the first line reduces to
\begin{align*}
\epsilon\sum_{j\in A}  \sum_{k\in B\setminus A}\sum_{l\in (B\setminus A)\setminus\{k\}}\sum_{i\in A\setminus\{j\}}
L_{jk}O_le_je_ke_l e_A  e_B  = \ & \epsilon(|A|-1)\sum_{j\in A}   \sum_{k\in B\setminus A}\sum_{l\in (B\setminus A)\setminus\{k\}}
L_{jk}O_le_je_ke_l e_A  e_B ,
\end{align*}
while, using Lemma~\ref{lemmA} for the last line, we find
\begin{align*}
&-\sum_{j\in A}  \sum_{k\in B\setminus A}\sum_{i\in A\setminus\{j\}}\sum_{l\in (B\setminus A)\setminus\{k\}}
L_{jk}e_iO_le_l e_ie_je_ke_A  e_B\\
= \ 
&\sum_{j\in A}  \sum_{k\in B\setminus A}\sum_{i\in A\setminus\{j\}}
L_{jk}e_iO_ie_ie_ie_je_ke_A  e_B+\sum_{j\in A}  \sum_{k\in B\setminus A}\sum_{i\in A\setminus\{j\}}
L_{jk}e_iO_je_je_ie_je_ke_A  e_B \\
&-\sum_{j\in A}  \sum_{k\in B\setminus A}\sum_{i\in A\setminus\{j\}}\sum_{l\in \{i,j\}\cup (B\setminus A) \setminus\{k\}}
L_{jk}e_ie_ie_je_k e_A  e_Be_lO_l \\
= \ &\epsilon \sum_{j\in A}  \sum_{k\in B\setminus A}\sum_{i\in A\setminus\{j\}}
L_{jk}e_i(O_ie_j -O_je_i)e_ke_A  e_B -\epsilon\sum_{j\in A}  \sum_{k\in B\setminus A}\sum_{i\in A\setminus\{j\}}
L_{jk}e_je_k e_A  e_Be_iO_i\\
&-\epsilon(|A|-1)\sum_{j\in A}  \sum_{k\in B\setminus A}\sum_{l\in \{j\}\cup (B\setminus A) \setminus\{k\}}
L_{jk}O_le_le_je_k e_A  e_B\\
= \ &-\epsilon \sum_{j\in A}  \sum_{k\in B\setminus A}\sum_{i\in A\setminus\{j\}}
L_{jk}e_ie_jO_ie_ke_A  e_B -\epsilon(-1)^{|B|-|A|}\sum_{j\in A}  \sum_{k\in B\setminus A}\sum_{i\in A\setminus\{j\}}
L_{jk}e_je_k e_ie_A  e_BO_i\\
&-\epsilon(|A|-1)\sum_{j\in A}  \sum_{k\in B\setminus A}\sum_{l\in  (B\setminus A) \setminus\{k\}}
L_{jk}O_le_le_je_k e_A  e_B.
\end{align*}
In total, the first summation of \eqref{(e)} thus yields
\begin{equation}
\label{f}
-\epsilon \sum_{j\in A}  \sum_{k\in B\setminus A}\sum_{i\in A\setminus\{j\}}
L_{jk}e_ie_jO_ie_ke_A  e_B-(-1)^{|B|-|A|}\epsilon\sum_{j\in A}  \sum_{k\in B}\sum_{i\in A\setminus\{j\}}
L_{jk}e_je_k e_ie_A  e_BO_i.
\end{equation}

The final result is obtained following essentially the same strategy as used in the proof of Theorem~\ref{theoAlgOii} and Theorem~\ref{theoAdisjB}, now applied to the third summation of \eqref{(e)}:
\begin{align*}
& 
-\epsilon \sum_{ i \in A}\sum_{ j \in A\setminus\{i\}}\sum_{l\in B\setminus A  }[L_{ij},L_{jl}] e_ie_l e_A  e_B . 
\end{align*}
By Theorem~\ref{theoLAlg}, and using $\{e_i,O_j\} = [p_i,x_j] - \delta_{ij}$, the summand of this can be written as
\begin{equation}
-\epsilon L_{il} e_i e_le_A  e_B-\epsilon L_{il}\{e_j,O_j\}e_ie_l e_A  e_B+
\epsilon L_{ij}\{e_j,O_l\}e_ie_l e_A  e_B- \epsilon L_{lj}\{e_i,O_j\}e_ie_l e_A  e_B .\label{thirdsum}
\end{equation}
Here, the summation of the first term cancels out with
\[
\left\llbracket \epsilon \frac{|A|-1}{2}e_A,-\sum_{\{k, l\} \subset B} L_{kl} e_ke_l e_B \right \rrbracket_- 
=   \epsilon (|A|-1) \sum_{k\in A} \sum_{l \in B\setminus A}  L_{kl}  e_ke_l e_A  e_B  ,
\]
which, using property~\eqref{eAeB} and $A\subset B$, follows from
\begin{align*}
(-1)^{|A||B|-|A\cap B|}e_ke_l e_B e_A   =\ &  (-1)^{|A||B|-|A|+|A||B\setminus\{k,l\}|-|A\cap (B\setminus\{k,l\})|}  e_A e_ke_l  e_B \\
=\ &  (-1)^{|A|-|A\cap( B\setminus\{k,l\})|}  e_A e_ke_l  e_B.
\end{align*}
For the summation of the second and the fourth term of \eqref{thirdsum}, using $O_ie_j-O_je_i=e_iO_j-e_jO_i$, we have
\begin{align*}
& -\sum_{ i \in A}\sum_{ j \in A\setminus\{i\}}\sum_{l\in B\setminus A  }( \epsilon L_{il}O_je_je_i +L_{lj}O_j  )e_le_A  e_B
-\epsilon  \sum_{ i \in A}\sum_{ j \in A\setminus\{i\}}\sum_{l\in B\setminus A  }
(L_{il}e_j  +L_{lj}e_i)O_je_ie_le_A e_B \\
= \ &
-\sum_{ \{i,j\} \in A}\sum_{l\in B\setminus A  }
(\epsilon L_{il}O_je_je_i +L_{lj}O_j +\epsilon L_{jl}O_ie_ie_j +L_{li}O_i  ) e_le_A  e_B
\\*
&-\epsilon  \sum_{ \{i,j\} \in A}\sum_{l\in B\setminus A  }
(L_{il}e_j  +L_{lj}e_i)(O_je_i-O_ie_j)e_le_A e_B .
\\
= \ &
- \epsilon \sum_{ \{i,j\} \in A}\sum_{l\in B\setminus A  }
(L_{il}O_je_je_i
- L_{lj}O_ie_ie_j -L_{il}e_j e_iO_j+L_{lj}e_ie_jO_i  )e_le_A  e_B\\
= \ 
& \epsilon\sum_{ i \in A}\sum_{ j \in A\setminus\{i\}} \sum_{l\in B\setminus A  }
L_{lj}(O_ie_ie_j-e_ie_jO_i  )e_le_A  e_B. \addtocounter{equation}{1}\tag{\theequation}\label{g}
\end{align*}
When summed over $l$ in $B\setminus A$, the third term of \eqref{thirdsum} yields, using Lemma~\ref{lemmA},
\begin{align*}
& -\epsilon \sum_{l\in B\setminus A  }L_{ij} e_jO_le_le_ie_A  e_B  +\epsilon \sum_{l\in B\setminus A  }L_{ij}  O_le_le_je_ie_A  e_B  \\
= \ & \epsilon L_{ij} e_jO_ie_ie_ie_A  e_B
-\epsilon \sum_{l\in B\setminus A \cup\{i\} }L_{ij} e_je_ie_A  e_Be_lO_l  \\
&-\epsilon L_{ij} O_ie_ie_je_ie_A  e_B -\epsilon L_{ij} O_je_je_je_ie_A  e_B +\epsilon \sum_{l\in B\setminus A \cup\{i,j\} }L_{ij} e_je_ie_A  e_Be_lO_l  \\
= \ & L_{ij} e_iO_je_A  e_B+ (-1)^{|B|-|A|+1} 
L_{ij} e_ie_A  e_BO_j .
\end{align*}
Summing this over $i$ in $A$ and $j$ in $A\setminus\{i\}$ and combined with
\eqref{(c)}, \eqref{(d)}, \eqref{f}, and \eqref{g}, one observes that all terms cancel out.

\end{proof}


\begin{proof}[Proof of Theorem~\ref{theoOAOB}]
By definition, plugging in \eqref{OA}, the left-hand side expands to
\begin{align*}
\llbracket O_{A},O_{B}\rrbracket_+  = \ & O_AO_B + (-1)^{|A||B|-|A\cap B|} O_BO_A \\
= \ &\frac14 \big(\uD\, \ux_A e_A - e_A \ux_A \uD - \epsilon e_A \big)\big(\uD\, \ux_B e_B - e_B \ux_B  \uD - \epsilon e_B \big) \\
&  + (-1)^{|A||B|-|A\cap B|} \frac14  \big(\uD\, \ux_B e_B - e_B \ux_B  \uD - \epsilon e_B \big)\big(\uD\, \ux_A e_A - e_A\ux_A   \uD - \epsilon e_A \big).
\end{align*}
The idea of the proof is now as follows. We split up $\ux_{A}$ and $\ux_{B}$ into $\ux_{A} =  \ux_{A\setminus B} +\ux_{A\cap B} $	and $\ux_{B} =  \ux_{B\setminus A} +\ux_{A\cap B} $. We then combine the appropriate terms 
to make  $\ux_{A\triangle B} =  \ux_{A\setminus B} + \ux_{B\setminus A}  $ and $\ux_{A\cup B} =  \ux_{A\setminus B} + \ux_{B\setminus A} +\ux_{A\cap B} $, and in turn all terms that make up the right-hand side. In doing so, we continually make use of the following facts: by property~\eqref{eA2} we have
$(e_{A\cap B})^4=1$, hence, $(e_{A\cap B})^2$ is just a sign; 
\[
\uD e_A - (-1)^{|A|} e_A \uD = \sum_{a\in A} 2p_a e_a e_A\,;
\]
and for integer $n$
one has $(-1)^{n}=(-1)^{n^2}$ and $(-1)^{n(n+1)}=1$.

For the terms of $\llbracket O_{A},O_{B}\rrbracket_+  $ which do not contain $\uD$, we have	
\begin{align*}
& \frac14 e_A e_B   + (-1)^{|A||B|-|A\cap B|}  \frac14 e_Be_A = \frac12 e_Ae_B = -\frac12 e_Ae_B+\frac12 (e_{A\cap B})^4e_Ae_B+\frac12 (e_{A\cap B})^4e_Ae_B.
\end{align*}

Next, for the terms of $\llbracket O_{A},O_{B}\rrbracket_+  $ containing two occurrences of $\uD$, we have	
\begin{align*}
& \big(\uD\, \ux_A e_A -e_A  \ux_A  \uD \big)\big(\uD\, \ux_B e_B - e_B \ux_B  \uD \big)   + (-1)^{|A||B|-|A\cap B|}  \big(\uD\, \ux_B e_B-  e_B \ux_B \uD\big)\big(\uD\, \ux_A e_A - e_A\ux_A   \uD  \big)\\
& 	= \ \uD\, \ux_A e_A\uD\, \ux_B e_B - e_A\ux_A   \uD\uD\, \ux_B e_B -  \uD\, \ux_A e_A e_B\ux_B \uD + e_A\ux_A   \uD e_B\ux_B \uD \\
&  \ \  + (-1)^{|A||B|-|A\cap B|} \big( \uD\, \ux_B e_B\uD\, \ux_A e_A - e_B\ux_B   \uD\uD\, \ux_A e_A- \uD\, \ux_Be_B e_A \ux_A  \uD  
+ e_B\ux_B   \uD  e_A\ux_A  \uD  \big).
\end{align*}	
We first look at the terms having $\uD\uD$ in the middle:
\begin{align*}	
&  - e_A\ux_A   \uD\uD\, \ux_B e_B-
(-1)^{|A||B|-|A\cap B|} e_B\ux_B   \uD\uD\, \ux_A e_A\\ 
= \ &
- e_A( \ux_{A\setminus B} +\ux_{A\cap B}  ) \uD\uD\, (\ux_{B\setminus A} +\ux_{A\cap B} ) e_B\\
&-(-1)^{|A||B|-|A\cap B|}
e_B(\ux_{B\setminus A} +\ux_{A\cap B} )  \uD\uD\,( \ux_{A\setminus B} +\ux_{A\cap B}  ) e_A\\ 
= \ & -(e_{A\cap B})^2\big(e_A e_{A\cap B}\ux_{A\setminus B}  \uD\uD\, \ux_{B\setminus A} e_{A\cap B} e_B\\
&-(-1)^{|A||B|-|A\cap B|+|A\cap B|(|A|+|B|)}
e_{A\cap B}e_B\ux_{B\setminus A}  \uD\uD\, \ux_{A\setminus B}  e_A e_{A\cap B}\\
&
-
(-1)^{1+|A\cap B|(|A|+|B|-2|A\cap B|+1)+|A\cap B|-1}e_{A\cap B}e_Ae_B \ux_{A\setminus B} \uD\uD\, \ux_{A\cap B}  e_{A\cap B}\\
&-
e_{A\cap B}\ux_{A\cap B}   \uD\uD\, \ux_{B\setminus A} e_{A\cap B}e_A  e_B
-
e_{A\cap B}\ux_{A\cap B}   \uD\uD\, \ux_{A\cap B}  e_{A\cap B}e_Ae_B\\
&-(-1)^{1+|A\cap B|(|A|+|B|-2|A\cap B|+1)+|A\cap B|-1}
e_{A\cap B}e_Ae_B\ux_{B\setminus A}   \uD\uD\,\ux_{A\cap B} e_{A\cap B}  \\
&-
e_{A\cap B} \ux_{A\cap B}   \uD\uD\, \ux_{A\setminus B} e_{A\cap B} e_Ae_B-(-1)^{|A\cap B|(|A|+|B|-2|A\cap B|)}
e_{A\cap B} e_Ae_B\ux_{A\cap B}   \uD\uD\,\ux_{A\cap B}  e_{A\cap B}\big)\\
= \ & -(e_{A\cap B})^2\big(e_Ae_{A\cap B}  \ux_{A\setminus B}  \uD \uD\, \ux_{B\setminus A} e_{A\cap B}e_B - e_{A\cap B}  \ux_{{A\cap B}}  \uD\uD\, \ux_{A\cup B} e_{A\cap B}e_Ae_B  \\
&-(-1)^{(|A|-|A\cap B|)(|B|-|A\cap B|) }
e_{A\cap B}e_B  \ux_{B\setminus A}  \uD \uD\, \ux_{A\setminus B} e_Ae_{A\cap B} \\
&-(-1)^{|A\cap B|(|A|+|B|-|A\cap B|)-|A\cap B| }e_{A\cap B}e_Ae_B  \ux_{A\cup B}  \uD\uD\, \ux_{A\cap B} e_{A\cap B} \big).
\end{align*}	

In exactly the same manner, we have
\begin{align*}	
&  \uD\, \ux_A e_Ae_B\ux_B   \uD
+(-1)^{|A||B|-|A\cap B|} \uD\, \ux_B e_Be_A\ux_A    \uD \\ 
= \ 	&  \uD\, ( \ux_{A\setminus B} +\ux_{A\cap B}  ) e_Ae_B(\ux_{B\setminus A} +\ux_{A\cap B} )   \uD\\
& +
(-1)^{|A||B|-|A\cap B|} \uD\, (\ux_{B\setminus A} +\ux_{A\cap B} ) e_Be_A( \ux_{A\setminus B} +\ux_{A\cap B}  )    \uD \\ 
= \ & (e_{A\cap B})^2\big(	\uD\, \ux_{A\setminus B} e_A(e_{A\cap B})^2e_B  \ux_{B\setminus A}  \uD 
-(-1)^{(|A|-|A\cap B|)(|B|-|A\cap B|) }	\uD\, \ux_{B\setminus A} e_{A\cap B}e_Be_Ae_{A\cap B}  \ux_{A\setminus B}  \uD 
\\
& -\uD\, \ux_{A\cap B} (e_{A\cap B})^2e_Ae_B  \ux_{A\cup B}  \uD   
-(-1)^{|A\cap B|(|A|+|B|-|A\cap B|)-|A\cap B| }
\uD\, \ux_{A\cup B} e_{A\cap B}e_Ae_Be_{A\cap B}  \ux_{{A\cap B}}  \uD\big)
\end{align*}	

Next, for the remaining terms with two occurrences of $\uD$, we first have
\begin{align*}	
& \uD\, \ux_A e_A\uD\, \ux_B e_B + (-1)^{|A||B|-|A\cap B|}  \uD\, \ux_B e_B\uD\, \ux_A e_A 
\\
= \ & \uD\, ( \ux_{A\setminus B} +\ux_{A\cap B}) e_A\uD\, (\ux_{B\setminus A} +\ux_{A\cap B}) e_B \\*
&
+ (-1)^{|A||B|-|A\cap B|}  \uD\,  (\ux_{B\setminus A} +\ux_{A\cap B})e_B\uD\, ( \ux_{A\setminus B} +\ux_{A\cap B}) e_A \\
= \ &  (e_{A\cap B})^2 \big(\uD\,  \ux_{A\setminus B}  e_Ae_{A\cap B}\uD \, \ux_{B\setminus A}e_{A\cap B}  e_B
\\*
&
+ (-1)^{|A\cap B|(|A|+|B|-|A\cap B|-1)}  \uD\, \ux_{A\setminus B} e_{A\cap B}e_Ae_B\uD\, \ux_{A\cap B} e_{A\cap B}\\*
&
+\uD\, \ux_{A\cap B} e_{A\cap B}\uD\, \ux_{B\setminus A} e_{A\cap B}e_A e_B
+\uD\, \ux_{A\cap B}e_{A\cap B}\uD\, \ux_{A\cap B} e_{A\cap B} e_Ae_B\\*
&
+ (-1)^{|A||B|-|A\cap B|+|A\cap B|(|B|+|A|)}  \uD\,  \ux_{B\setminus A} e_{A\cap B}e_B\uD\,  \ux_{A\setminus B} e_Ae_{A\cap B}\\*
&+ (-1)^{|A\cap B|(|A|+|B|)} \uD\,  \ux_{B\setminus A} e_{A\cap B}e_Ae_B\uD\, \ux_{A\cap B} e_{A\cap B}
+ \uD\, \ux_{A\cap B}e_{A\cap B}\uD\,  \ux_{A\setminus B} e_{A\cap B} e_Ae_B\\*
&
+ (-1)^{|A\cap B|(|A|+|B|)}  \uD\, \ux_{A\cap B}e_{A\cap B}e_Ae_B\uD\,  \ux_{A\cap B} e_{A\cap B}\big)\\*
&- \uD\,  \ux_{A\setminus B} \sum_{l\in A\cap B} 2p_le_l \ux_{B\setminus A}  e_A e_B
-  \uD\,  \ux_{A\setminus B}  \sum_{l\in B\setminus A} 2p_le_l \ux_{A\cap B} e_A e_B- \uD\,  \ux_{A\cap B} \sum_{l\in A\setminus B} 2p_le_l \ux_{B\setminus A}  e_A e_B\\*
&
+  \uD\,  \ux_{A\cap B}  \sum_{l\in A\setminus B} 2p_le_l \ux_{ A\cap B}  e_Ae_B
-  \uD\,  \ux_{B\setminus A}  \sum_{l\in A\cap B} 2p_le_l \ux_{A\setminus B}  e_Ae_B
-  \uD\,  \ux_{B\setminus A}  \sum_{l\in A\setminus B} 2p_le_l \ux_{A\cap B}  e_Ae_B\\*
&
-  \uD\,  \ux_{A\cap B}  \sum_{l\in B\setminus A} 2p_le_l\, \ux_{A\setminus B}  e_Ae_B
- \uD\,  \ux_{A\cap B}  \sum_{l\in A\setminus B} 2p_le_l\, \ux_{A\cap B}  e_Ae_B \\
= \ &  (e_{A\cap B})^2 \big(\uD\,  \ux_{A\setminus B}  e_Ae_{A\cap B}\uD \, \ux_{B\setminus A}e_{A\cap B}  e_B
+   \uD\, \ux_{A\cap B}e_{A\cap B}\uD\,  \ux_{A\cup B} e_{A\cap B} e_Ae_B\\*
&
+ (-1)^{(|A|-|A\cap B|)(|B|-|A\cap B|) }  \uD\,  \ux_{B\setminus A} e_{A\cap B}e_B\uD\,  \ux_{A\setminus B} e_Ae_{A\cap B}\\*
&
+  (-1)^{|A\cap B|(|A|+|B|-|A\cap B|)-|A\cap B| }  \uD\, \ux_{A\cup B}e_{A\cap B}e_Ae_B\uD\,  \ux_{A\cap B} e_{A\cap B}\big)\\*
&- 2\uD\, \sum_{a\in A\setminus B} \sum_{c\in A\cap B}\sum_{b\in B\setminus A}\big(- x_a (p_c  x_b 
-    p_b x_c ) 
+   x_c (  p_a  x_b  -    p_b x_a )
+   x_b  (  p_c x_a 
-     p_a x_c )
\big) e_ae_be_ce_Ae_B
\end{align*}	
where the last line vanishes using  $L_{ij}=x_ip_j-x_jp_i = p_jx_i-p_ix_j$. Similarly one finds
\begin{align*}	
&  e_A\ux_A   \uD e_B\ux_B \uD + (-1)^{|A||B|-|A\cap B|}  e_B\ux_B   \uD  e_A\ux_A  \uD  \\
= \ & e_A ( \ux_{A\setminus B} +\ux_{A\cap B}) \uD\,e_B (\ux_{B\setminus A} +\ux_{A\cap B}) \uD\, \\
&
+ (-1)^{|A||B|-|A\cap B|}   e_B (\ux_{B\setminus A} +\ux_{A\cap B})\uD\,e_A ( \ux_{A\setminus B} +\ux_{A\cap B}) \uD\, \\
= \ &  (e_{A\cap B})^2 \big( e_Ae_{A\cap B} \ux_{A\setminus B}  \uD \,e_{A\cap B}  e_B \ux_{B\setminus A}\uD\,	+   e_{A\cap B} \ux_{A\cap B}\uD\, e_{A\cap B} e_Ae_B \ux_{A\cup B} \uD\,
\\
&
+ (-1)^{(|A|-|A\cap B|)(|B|-|A\cap B|) } e_{A\cap B}e_B \ux_{B\setminus A} \uD\,  e_Ae_{A\cap B} \ux_{A\setminus B} \uD\, 
\\
&
+  (-1)^{|A\cap B|(|A|+|B|-|A\cap B|)-|A\cap B| }   e_{A\cap B}e_Ae_B\ux_{A\cup B}\uD\, e_{A\cap B} \ux_{A\cap B} \uD\big).
\end{align*}	

Finally, for the terms with a single occurrence of $\uD$, we have (up to a factor $-\epsilon$)
\begin{align*}
&  \big(\uD\, \ux_A e_A-e_A\ux_A   \uD  \big) e_B  + (-1)^{|A||B|-|A\cap B|}  e_B \big(\uD\, \ux_A e_A -  e_A \ux_A \uD  \big)\\
& 
+  e_A \big(\uD\, \ux_B e_B - e_B\ux_B   \uD  \big)   + (-1)^{|A||B|-|A\cap B|}  \big(\uD\, \ux_B e_B - e_B\ux_B   \uD \big) e_A  \\
= \ & - \uD\, \ux_{A\setminus B}  e_Ae_B - \uD\, \ux_{B\setminus A}  e_Ae_B  +   e_Ae_B \ux_{A\setminus B}\uD +    e_Ae_B\ux_{B\setminus A}\uD\\
&+(e_{A\cap B})^2 \big( \uD\, \ux_{A\setminus B}  e_A(e_{A\cap B})^2e_B
+(-1)^{|A\cap B|(|A|+|B|)} \uD\, \ux_{A\setminus B} e_{A\cap B} e_Ae_Be_{A\cap B} \\
&
+ \uD\, \ux_{A\cap B}(e_{A\cap B})^2 e_A e_B 
-e_A e_{A\cap B}\ux_{A\setminus B}\uD  e_{A\cap B}  e_B 
-e_{A\cap B} \ux_{A\cap B} \uD  e_{A\cap B}e_A e_B\\
&
+ (-1)^{|A||B|+|A\cap B|(|B|+|A|-1)} e_{A\cap B} e_B \uD\, \ux_{A\setminus B} e_Ae_{A\cap B}  + (-1)^{|A\cap B|(|A|+|B|)}  e_{A\cap B}e_A e_B\uD\,  \ux_{A\cap B}e_{A\cap B}
\\ 
& -(-1)^{|A||B|+|A\cap B|(|B|+|A|-1)}e_{A\cap B}e_B e_Ae_{A\cap B}\ux_{A\setminus B}\uD 
-  (-1)^{|A\cap B|(|A|+|B|)} e_{A\cap B}e_Ae_B\ux_{A\setminus B}  \uD\,e_{A\cap B}\\
&    -  (e_{A\cap B})^2e_Ae_B\ux_{A\cap B}  \uD
+  e_A e_{A\cap B}\uD\, \ux_{B\setminus A} e_{A\cap B}e_B 
+ e_{A\cap B} \uD\, \ux_{A\cap B} e_{A\cap B} e_Ae_B \\*
&   
-   e_A(e_{A\cap B})^2e_B\ux_{B\setminus A}\uD -(e_{A\cap B})^2 e_Ae_B \ux_{B\setminus A}   \uD 
- (-1)^{|A\cap B|(|A|+|B|)}e_{A\cap B}e_Ae_Be_{A\cap B} \ux_{A\cap B}  \uD  \\*
& + (-1)^{|A\cap B|(|A|+|B|)}\uD\, \ux_{B\setminus A} e_{A\cap B} e_Ae_B e_{A\cap B}
+ (-1)^{|A||B|+|A\cap B|(|A|+|B|-1)}\uD\, \ux_{B\setminus A} e_{A\cap B} e_Be_A e_{A\cap B}\\*
&   + (-1)^{|A\cap B|(|A|+|B|)}\uD\,\ux_{A\cap B} e_{A\cap B}e_Ae_Be_{A\cap B}
- (-1)^{|A||B|+|A\cap B|(|B|+|A|-1)}  e_{A\cap B} e_B\ux_{B\setminus A} \uD e_A  e_{A\cap B} \\*
& 
- (-1)^{|A\cap B|(|A|+|B|)} e_{A\cap B} e_A e_B\ux_{A\cap B}  \uD  e_{A\cap B} \big), 
\end{align*}		
where, for instance, we made use of the following computation
\begin{align*}
& -e_A\ux_{A\cap B}   \uD  e_B  
- (-1)^{|A||B|-|A\cap B|}   e_B\ux_{A\cap B}   \uD  e_A  \\
=\ & -(e_{A\cap B})^2 \big((e_{A\cap B})^2e_A\ux_{A\cap B}   \uD  e_B  
+ (-1)^{|A||B|-|A\cap B|}   e_B\ux_{A\cap B}   \uD  (e_{A\cap B})^2e_A \big) \\
=\ & -(e_{A\cap B})^2 \big( (-1)^{|A|-|A\cap B|} e_{A\cap B}\ux_{A\cap B}   e_{A\cap B}e_A\uD  e_B  
+ (-1)^{|A||B|+|A\cap B||A|}   e_B\ux_{A\cap B}   \uD  e_{A\cap B}e_Ae_{A\cap B} \big) \\
=\ & -(e_{A\cap B})^2 \big(  e_{A\cap B}\ux_{A\cap B}  \uD e_{A\cap B}e_A  e_B  
+ (-1)^{|A||B|+|A\cap B||A|+|A|-|A\cap B|}   e_B\ux_{A\cap B}  e_{A\cap B}e_A \uD  e_{A\cap B}  \\
& - e_{A\cap B}\ux_{A\cap B} \sum_{l\in A\setminus B}2p_le_l  e_{A\cap B} e_A e_B  + \,(-1)^{|A||B|+|A\cap B||A|} e_B\ux_{A\cap B} \sum_{l\in A\setminus B}2p_le_l  e_{A\cap B}e_A e_{A\cap B} \big)\\
=\ &(e_{A\cap B})^2 \big( -e_{A\cap B} \ux_{A\cap B} \uD  e_{A\cap B}e_A e_B
- (-1)^{|A\cap B|(|A|+|B|)} e_{A\cap B} e_A e_B\ux_{A\cap B}  \uD  e_{A\cap B} \big).
\end{align*}

Putting everything together and comparing with 
\[
O_{A, B}  =  \frac{1}{2}\big(\uD\, \ux_{A\triangle B} e_Ae_B -e_Ae_B  \ux_{A\triangle B}  \uD - \epsilon e_Ae_B \big),
\]	
and	
\begin{align*}
&	\llbracket O_{A, (A\cap B)},O_{(A\cap B),  B}\rrbracket_+ = \frac{1}{4}\big(\uD\, \ux_{A\setminus B} e_Ae_{A\cap B} -e_Ae_{A\cap B}  \ux_{A\setminus B}  \uD - \epsilon e_Ae_{A\cap B} \big)\\
& \times
\big(\uD\, \ux_{B\setminus A} e_{A\cap B}e_B -e_{A\cap B}e_B  \ux_{B\setminus A}  \uD - \epsilon e_{A\cap B}e_B \big) +(-1)^{(|A|-|A\cap B|)(|B|-|A\cap B|) }\frac{1}{4}\\
&\times
\big(\uD\, \ux_{B\setminus A} e_{A\cap B}e_B -e_{A\cap B}e_B  \ux_{B\setminus A}  \uD - \epsilon e_{A\cap B}e_B \big)
\big(\uD\, \ux_{A\setminus B} e_Ae_{A\cap B} -e_Ae_{A\cap B}  \ux_{A\setminus B}  \uD - \epsilon e_Ae_{A\cap B} \big),
\end{align*}
and
\begin{align*}
&  \llbracket O_{A\cap B},O_{(A\cap B) ,A, B}\rrbracket_+ 
= \frac{1}{4}\big(\uD\, \ux_{A\cap B} e_{A\cap B} -e_{A\cap B}  \ux_{{A\cap B}}  \uD - \epsilon e_{A\cap B} \big)\\
& \times
\big(\uD\, \ux_{A\cup B} e_{A\cap B}e_Ae_B -e_{A\cap B}e_Ae_B  \ux_{A\cup B}  \uD - \epsilon e_{A\cap B}e_Ae_B \big) +(-1)^{|A\cap B|(|A|+|B|-|A\cap B|)-|A\cap B| }\frac{1}{4}\\
& \times\big(\uD\, \ux_{A\cup B} e_{A\cap B}e_Ae_B -e_{A\cap B}e_Ae_B  \ux_{A\cup B}  \uD - \epsilon e_{A\cap B} e_Ae_B \big)
\big(\uD\, \ux_{A\cap B} e_{A\cap B} -e_{A\cap B}  \ux_{A\cap B} \uD - \epsilon e_{A\cap B} \big) ,
\end{align*}
the proof is completed.	
\end{proof}

\section*{Acknowledgments}

The research of HDB is supported by the Fund for Scientific
Research-Flanders (FWO-V), project ``Construction of algebra realizations
using Dirac-operators,'' grant G.0116.13N.

\end{document}